\documentclass[11pt]{article}
\usepackage[square,sort,comma,numbers]{natbib}
\usepackage{amsmath}
\usepackage{amsthm}
\usepackage{amssymb}
\usepackage{mathtools}
\usepackage{algorithm}
\usepackage{subfig}
\usepackage{color}
\usepackage[english]{babel}
\usepackage{graphicx}
\usepackage{natbib}
\usepackage{wrapfig,epsfig}
\usepackage{epstopdf}
\usepackage{url}
\usepackage{graphicx}
\usepackage{color}
\usepackage{algpseudocode}
\usepackage[pdfencoding=auto,psdextra]{hyperref}
\hypersetup{
    colorlinks = true,
    citecolor = blue,
    linkcolor = black
}
\usepackage[T1]{fontenc}
\usepackage{bbm}
\usepackage{dsfont}

\usepackage{tikz}
\usetikzlibrary{arrows}
\usepackage[margin=1in]{geometry} 
\usepackage{nccmath}
\DeclareMathOperator*{\J}{\mathcal{J}}
\DeclareMathOperator*{\M}{\mathcal{M}}
\DeclareMathOperator*{\suchthat}{\,|\,}
\DeclareMathOperator*{\OPT}{\operatorname{OPT}}

\usepackage{thmtools}
\usepackage{thm-restate}

\usepackage{tikz}
\usetikzlibrary{arrows}
\usetikzlibrary{decorations.markings}
\usetikzlibrary{arrows}
\usepackage{graphicx}
\usepackage{pgfplots}
\usepgfplotslibrary{fillbetween}

\def\eps{\varepsilon}

\newcommand\blfootnote[1]{%
  \begingroup
  \renewcommand\thefootnote{}\footnote{#1}%
  \addtocounter{footnote}{-1}%
  \endgroup
}

\author{
Dimitris Fotakis\\
\texttt{fotakis@cs.ntua.gr}\\
School of Electrical and Computer Engineering,\\
National Technical University of Athens, Greece\\
\and
Jannik Matuschke\\
\texttt{jannik.matuschke@kuleuven.be}\\
Faculty of Economics and Business\\
KU Leuven, Belgium\\
\and
Orestis Papadigenopoulos\\
\texttt{papadig@cs.utexas.edu}\\
Department of Computer Science\\
The University of Texas at Austin, USA.
}

\newtheorem{theorem}{Theorem}[section]

\newtheorem{proposition}[theorem]{Proposition}

\newtheorem{fact}[theorem]{Fact}
\newtheorem{remark}[theorem]{Remark}

\title{Malleable scheduling beyond identical machines\blfootnote{Part of this work was carried out while the authors participated in the program ``Real-Time Decision Making'' at the Simons Institute for the Theory of Computing, Berkeley, CA. Part of this work was carried out while the second author worked at Technische Universit\"at M\"unchen and was supported by the Alexander von Humboldt Foundation with funds of the German Federal Ministry of Education and Research (BMBF).}}
\date{\today}

\begin{document} 
\pagenumbering{gobble}
\begin{titlepage}
  \maketitle
\begin{abstract}
In malleable job scheduling, jobs can be executed simultaneously on multiple machines with the processing time depending on the number of allocated machines. In this setting, jobs are required to be executed non-preemptively and in unison, in the sense that they occupy, during their execution, the same time interval over all the machines of the allocated set. In this work, we study generalizations of malleable job scheduling inspired by standard scheduling on unrelated machines. Specifically, we introduce a general model of malleable job scheduling, where each machine has a (possibly different) speed for each job, and the processing time of a job $j$ on a set of allocated machines $S$ depends on the total speed of $S$ with respect to $j$. For machines with unrelated speeds, we show that the optimal makespan cannot be approximated within a factor less than $\frac{e}{e-1}$, unless $P = NP$. On the positive side, we present polynomial-time algorithms with approximation ratios $\frac{2e}{e-1}$ for machines with unrelated speeds, $3$ for machines with uniform speeds, and $7/3$ for restricted assignments on identical machines. Our algorithms are based on deterministic LP rounding. They result in sparse schedules, in the sense that each machine shares at most one job with other machines. We also prove lower bounds on the integrality gap of $1+\varphi$ for unrelated speeds ($\varphi$ is the golden ratio) and $2$ for uniform speeds and restricted assignments. To indicate the generality of our approach, we show that it also yields constant factor approximation algorithms for a variant where we determine the effective speed of a set of allocated machines based on the $L_p$ norm of their speeds. 
\end{abstract}
\end{titlepage}

\clearpage
\pagenumbering{arabic}

\section{Introduction}
\label{sec:intro}

Since the late 60s, various models have been proposed by researchers \cite{GG75,G69} in order to capture the real-world aspects and particularities of multiprocessor task scheduling systems, i.e., large collections of identical processors able to process tasks in parallel. High performance computing, parallel architectures, and cloud services are typical applications that motivate the study of multiprocessor scheduling, both theoretical and practical. An influential model is Rayward-Smith's unit execution time and unit communication time (UET-UCT) model~\cite{R87}, where each parallel job is partitioned into a set of tasks of unit execution time and these tasks are subject to precedence constraints modeled by a task graph.
The UET-UCT model and its generalizations have been widely studied and a large number of (approximation) algorithms and complexity results have been proposed \cite{HM01,PY90}. 

However, the UET-UCT model mostly focuses on task scheduling and sequencing, and does not account for the amount of resources allocated to each job, thus failing to capture an important aspect of real-world parallel systems. Specifically, in the UET-UCT model, the level of granularity of a job (that is, the number of smaller tasks that a job is partitioned into) is decided a priori and 
is given as part of the input. 
However, it is common ground in the field of parallel processing that the unconditional allocation of resources for the execution of a job may jeopardize the overall efficiency of a multiprocessor system. A theoretical explanation is provided by Amdahl's law~\cite{A07}, which suggests that the speedup of a job's execution can be estimated by the formula $\frac{1}{(1-p) + \frac{p}{s}}$, where $p$ is the fraction of the job that can be parallelized and $s$ is the speedup due to parallelization (i.e., $s$ can be thought of as the number of processors). 

\smallskip\noindent{\bf Malleable Scheduling.}
An interesting alternative to the UET-UCT model is that of {\em malleable}\footnote{Malleable scheduling also appears as {\em moldable}, while sometimes the two terms refer to slightly different models.} job scheduling~\cite{DL89,TWY92}. In this setting, a set $\J$ of jobs is scheduled on a set $\M$ of parallel machines, while every job can be processed by more than one machine at the same time. In order to quantify the effect of parallelization, the processing time of a job $j \in \J$ is determined by a function $f_j : \mathbb{N} \rightarrow \mathbb{R}_+$ depending on the number of allocated machines.\footnote{We denote by $\mathbb{R}_+$ (resp. $\mathbb{Z}_+$) the set of non-negative reals (resp. integers).}  Moreover, every job must be executed {\em non-preemptively} and in {\em unison}, i.e. having the same starting and completion time on each of the allocated machines. Thus, if a job $j$ is assigned to a set of machines $S$ starting at time $\tau$, all machines in $S$ are occupied with job $j$ during the interval $[\tau, \tau+f_j(|S|)]$. It is commonly assumed that the processing time function of a job exhibits two useful and well-motivated properties: 
\begin{itemize}
\item For every job $j \in \J$, the processing time $f_j(s)$ is {\em non-increasing} in the number of machines.\footnote{This property holds w.l.o.g., as the system always has the choice not to use some of the allocated machines.}
\item The total {\em work} of the execution of a job $j$ on $s$ machines, that is the product $s \cdot f_j(s)$, is {\em non-decreasing} in the number of machines.
\end{itemize}
The latter property, known as {\em monotonicity} of a malleable job, is justified by Brent's law~\cite{B74}: One cannot expect superlinear speedup by increasing the level of parallelism. 
A great deal of theoretical results have been published on scheduling malleable jobs according to the above model (and its variants) for the objective of minimizing the {\em makespan}, i.e., the completion time of the last finishing job, or other standard objectives (see e.g., \cite{DMT04} and the references therein).
 
Although malleable job scheduling represents a va\-liant attempt to capture real-world aspects of massively parallel processing, the latter exhibits even more complicated characteristics. Machine heterogeneity, data locality and hardware interconnection are just a few aspects of real-life systems that make the generalization of the aforementioned model necessary. In modern multiprocessor systems, machines are not all identical and the processing time of a job not only depends on the quantity, but also on the quality of the set of allocated machines. Indeed, different physical machines may have different capabilities in terms of faster CPUs or more efficient cache hierarchies. Moreover, the above heterogeneity may be job-dependent, in the sense that a specific machine may be faster when executing a certain type of jobs than another (e.g. memory- vs arithmetic-intensive applications \cite{PH13}). Finally, the execution of a job on specific combinations of machines may also yield additional benefit (e.g., machines that are local in terms of memory hierarchy). 

\smallskip\noindent{\bf Our Model: Malleable Scheduling on Unrelated Machines.}
Quite surprisingly, no results exist on sche\-duling malleable jobs beyond the case of identical machines, to the best of our knowledge, despite the significant theoretical and practical interest in the model. In this work, we extend the model of malleable job scheduling to capture more delicate aspects of parallel job scheduling. In this direction, while we still require our jobs to be executed non-preemptively and in unison, the processing time of a job $j \in \J$ becomes a set function $f_j(S)$, where $S \subseteq \M$ is the set of allocated machines. We require that processing times are given by a {\em non-increasing} function, in the set function context, while additional assumptions on the scalability of $f_j$ are made, in order to capture the {\em diminishing utility} property implied by Brent's law. 

These assumptions naturally lead to a generalized malleable job setting, where processing times are given by non-increasing supermodular set functions $f_j(S)$, accessed by value queries. We show that makespan minimization in this general setting 
is inapproximable within $\mathcal{O}(|\J|^{1-\eps})$ factors (unless $P = NP$, see Section~\ref{sec:supermodular}). The general message of the proof is that unless we make some relatively strong assumptions on processing times (in the form e.g., of a relatively smooth gradual decrease in the processing time, as more machines are allocated), malleable job scheduling (even with monotone supermodular processing times) can encode combinatorial problems as hard as graph coloring. 

Thus, inspired by (standard non-malleable) scheduling models on uniformly related and unrelated machines, we introduce the notion of {\em speed-implementable} processing time functions. For each machine $i$ and each job $j$ there is a {\em speed} $s_{i,j} \in \mathbb{Z}_{+}$ that quantifies the contribution of machine $i$ to the execution of job $j$, if $i$ is included in the set allocated to $j$. For most of this work, we assume that the total speed of an allocated set is given by an additive function $\sigma_j(S) = \sum_{i \in S} s_{i,j}$ (but see also Section~\ref{sec:effective-speed}, where we discuss more general speed functions based on $L_p$-norms). 
A function is speed-implementable if we can write $f_j(S) = f_j(\sigma_j(S))$ for some function $f_j : \mathbb{R}_{+} \rightarrow \mathbb{R}_{+}$.\footnote{For convenience, we use the identifier $f_j$ for both functions. Since their arguments come from disjoint domains, it is always clear from the context which one is meant.}  Again, we assume oracle access to the processing time functions.

The notion of speed-implementable processing times allows us to quantify the fundamental assumptions of {\em monotonicity} and {\em diminishing utility} in a clean and natural way. More specifically, we make the following two assumptions on speed-implementable functions:
\begin{enumerate}
\item {\em Non-increasing processing time:} For every job $j \in \J$, the processing time $f_j(s)$ is non-increasing in the total allocated speed $s \in \mathbb{R}_+$. 
\item {\em Non-decreasing work:} For every job $j \in \J$, the work $f_j(s) \cdot s$ is non-decreasing in the total allocated speed $s \in \mathbb{R}_+$. 
\end{enumerate}
The first assumption ensures that allocating more speed cannot increase the processing time. The second assumption is justified by Brent's law, when the increase in speed coincides with an increase in the physical number of machines, or by similar arguments for the increase of the total speed of a single physical machine (e.g., memory access, I/O bottleneck \cite{PH13} etc.). We remark that speed-implementable functions with non-in\-creasing processing times and non-decreasing work do not need to be convex, and thus, do not belong to the class of supermodular functions. 

In order to avoid unnecessary technicalities, we additionally assume that for any job $j \in \J$, $f_j(0) = +\infty$, namely, no job can be executed on a set of machines of zero cumulative speed.

In this work, we focus on the objective of minimizing the \emph{makespan} $C_{max} = \max_{j \in \J} C_j$, where $C_j$ the completion time of job~$j$. 
We refer to this setting as the problem of {\em scheduling malleable jobs on unrelated machines}. To further justify this term, we present a pseudopolynomial transformation of standard scheduling on unrelated machines to malleable scheduling with speed-implementable processing times (see Section~\ref{sec:hardness}). The reduction can be rendered polynomial by standard techniques, preserving approximation factors with a loss of $1+\varepsilon$.

\subsection{Related Work}
The problem of malleable job scheduling on identical machines has been studied thoroughly for more than three decades.
For the case of non-monotonic jobs, i.e., jobs that do not satisfy the monotonic work condition, Du and Leung~\cite{DL89} show that the problem is strongly NP-hard for more than $5$ machines, while in terms of approximation, Turek, Wolf and Yu~\cite{TWY92} provided the first 2-approximation algorithm for the same version of the problem. 
Jansen and Porkolab~\cite{JP02} devised a PTAS for instances with a constant number of machines, which was later extended by Jansen and Th\"ole~\cite{JT10} to a PTAS for the case that the number of machines is polynomial in the number of jobs.

For the case of monotonic jobs, Mouni{\'{e}}, Rapine and Trystram~\cite{MRT07} propose a $\frac{3}{2}$-approximation algorithm, improving on the $\sqrt{3}$-approximation provided by the same authors~\cite{MRT99}. 
Recently, Jansen and Land~\cite{JL17} gave an FPTAS for the case that $|\M| \geq 8 |\J|/\varepsilon$. 
Together with the approximation scheme for polynomial number of machines in~\cite{JL17}, this implies a PTAS for scheduling monotonic malleable jobs on identical machines.

Several papers also consider the problem of scheduling malleable jobs with preemption and/or under precedence constraints~\cite{BKMTW06, JZ06, MP14}.
An interesting alternative approach to the general problem is that of Srinivasa Prasanna, and Musicus~\cite{SM91}, who consider a continuous version of malleable tasks and develop an exact algorithm based on optimal control theory under certain assumptions on the processing time functions.
While the problem of malleable scheduling on identical machines is very well understood, this is not true for malleable extensions of other standard scheduling models, such as unrelated machines or the restricted assignment model.

A scheduling model similar to malleable tasks is that of \emph{splittable jobs}. In this regime, jobs can be split arbitrarily and the resulting parts can be distributed arbitrarily on different machines. 
For each pair of job $j$ and machine $i$, there is a setup time $s_{ij}$ and a processing time $p_{ij}$.
If a fraction $x_{ij} \in (0, 1]$ of job $j$ is to be scheduled on machine $i$, the load that is incurred on the machine is $s_{ij} + p_{ij}x_{ij}$.
Correa et al.~\cite{CM15} provide an $(1 + \varphi)$-approximation algorithm for this setting (where $\varphi$ is the golden ratio), which is based on an adaptation of the classic LP rounding result by Lenstra, Shmoys, and Tardos~\cite{LST90} for the traditional unrelated machine scheduling problem.
We remark that the generalized malleable setting considered in this paper also induces a natural generalization of the splittable setting beyond setup times, when dropping the requirement that jobs need to be executed in unison.
As in~\cite{CM15}, we provide a rounding framework based on a variant of the assignment LP from~\cite{LST90}. However, the fact that processing times are only given implicitly as functions in our setting makes it necessary to carefully choose the coefficients of the assignment LP, in order to ensure a constant integrality gap. Furthermore, because jobs have to be executed in unison, we employ a more sophisticated rounding scheme in order to better utilize free capacity on different machines.

\subsection{Contribution and Techniques}
\label{sec:contrib}

At the conceptual level, we introduce the notion of malleable jobs with speed-implementable processing times. Hence, we generalize the standard and well-studied setting of malleable job scheduling, in a direct analogy to fundamental models in scheduling theory (e.g., scheduling on \emph{uniformly related} and \emph{unrelated} machines). This new and much richer model gives rise to a large family of unexplored packing problems that may be of independent interest. 

From a technical viewpoint, we investigate the computational complexity and the approximability of this new setting. To the best of our understanding, standard techniques used for makespan minimization in the setting of malleable job scheduling on identical machines, such as the two-shelve approach (as used in \cite{MRT07,TWY92}) and area charging arguments, fail to yield any reasonable approximation guarantees in our more general setting. 
This intuition is supported by the following hardness of approximation result (see Section~\ref{sec:hardness} for the proof).
\begin{restatable}{theorem}{restateHardness}\label{unrelated:apxhardness}
For any $\epsilon > 0$, there is no $(\frac{e}{e-1} - \epsilon)$-approximation algorithm for the problem of scheduling malleable jobs on unrelated machines, unless $P=NP$.
\end{restatable}
Note that the lower bound of $\frac{e}{e-1}$ is strictly larger than the currently best known lower bound of $1.5$ for classic (non-malleable) scheduling on unrelated machines.

Our positive results are based on a linear programming relaxation, denoted by [LP($C$)] and described in Section~\ref{sec:overview}. This LP resembles the assignment LP for the standard setting of non-malleable scheduling~\cite{LST90}. However, in order to obtain a constant integrality gap we distinguish between ``small'' jobs that can be processed on a single machine (within a given target makespan), and ``large'' jobs that have to be processed on multiple machines. 
For the large jobs, we carefully estimate their contribution to the load of their allocated machines. Specifically, we introduce the notion of \emph{critical speed} and use the critical speed to define the load coefficients incurred by large jobs on machines in the LP relaxation by proportionally distributing the work volume according to machine speeds. 
For the rounding, we exploit the sparsity of our relaxation's extreme points (as in \cite{LST90}) and generalize the approach of \cite{CM15}, in order to carefully distinguish between jobs assigned to a single machine and jobs shared by multiple machines.
\begin{restatable}{theorem}{restateUnrelated} \label{thm:filtering-approx}
There exists a polynomial-time $\frac{2 e}{e-1}$-appro\-xi\-mation algorithm for the problem of scheduling malleable jobs on unrelated machines.
\end{restatable}
An interesting corollary is that for malleable job scheduling on unrelated machines, there always exists an approximate solution where each machine shares at most one job with some other machines.

In addition, we consider the interesting special case of {\em restricted assignment} machines, namely, the case where each job is associated with subset $M_j \subseteq \M$ of machines such that $s_{i,j} = 1$ for all $j \in \J$ and $i \in M_j$ and $s_{i,j} = 0$ otherwise. Finally, we consider the case of {\em uniform speeds}, where $s_{i,j} = s_{i}$ for all $i \in \M$ and $j \in \J$. For the above special cases of our problem, we are able to get improved approximation guarantees by exploiting the special structure of the processing time functions, as summarized in the following two theorems.

\begin{restatable}{theorem}{restateRestricted}\label{rounding:restricted}
There exists a polynomial-time $\frac{7}{3}$-approxi\-mation algorithm for the problem of scheduling malleable jobs on restricted identical machines.
\end{restatable}

\begin{restatable}{theorem}{restateUniform}\label{uniform:thm}
There exists a polynomial-time $3$-approxi\-mation algorithm for the problem of scheduling malleable jobs on uniform machines.
\end{restatable}

All our approximation results imply corresponding upper bounds on the integrality gap of the linear programming relaxation [LP($C$)].
Based on an adaptation of a construction in~\cite{CM15}, we show a lower bound of $1+\varphi \approx 2.618$ on the integrality gap of [LP($C$)] for malleable job scheduling on unrelated machines, where $\varphi$ is the golden ratio (see Section~\ref{sec:hardness}).
For the cases of restricted assignment and uniformly related machines, respectively, we obtain an integrality gap of $2$.

Moreover, we extend our model and approach in the following direction. We consider a setting where the \emph{effective speed} according to which a set $S$ of allocated machines processing a job $j$ is given by the $L_p$-norm $\sigma^{(p)}_j(S) = \left( \sum_{i \in S} (s_{i,j})^p  \right)^{1/p}$ of the corresponding speed vector. In practical settings, we tend to prefer assignments to relatively small sets of physical machines, so as to avoid delays related to communication, memory access, and I/O (see e.g., \cite{PH13}). By replacing the total speed (i.e., the $L_1$-norm) with the $L_p$-norm of the speed vector for some $p \geq 1$, we discount the contribution of additional machines (especially of smaller speeds) towards processing a job $j$. Thus, as $p$ increases, we give stronger preference to \emph{sparse} schedules, where the number of jobs shared between different machines (and the number of machines sharing a job) are kept small. Interestingly, our general approach is robust to this generalization and it results in constant approximation factors for any $p \geq 1$. 
Asymptotically, the approximation factor is bounded by $\frac{p}{p-\ln p} + \sqrt[p]{\frac{p}{\ln p}}$ and our algorithm smoothly converges to the classic 2-approximation algorithm for unrelated machine scheduling~\cite{LST90} as $p$ tends to infinity (note that for the $L_\infty$-norm, our setting is identical to standard scheduling on unrelated machines).
These results are discussed in Section~\ref{sec:effective-speed}.

Trying to generalize malleable job scheduling beyond the simple setting of identical machines, as much as possible, we believe that our setting with speed-implementable processing times lies on the frontier of the constant-factor approximability regime. We show a strong inapproximability lower bound of $\mathcal{O}(| \J |^{1-\eps})$ for the (far more general) setting where the processing times are given by a non-increasing supermodular set functions. These results are discussed in Section~\ref{sec:supermodular}. An interesting open question is to characterize the class of processing time functions for which malleable job scheduling admits constant factor (and/or logarithmic) approximation guarantees. 
\section{The general rounding framework} \label{sec:overview}
In this section, we provide a high-level description of our algorithm. We construct a polynomial-time {\em $\rho$-relaxed decision procedure} for malleable job scheduling problems. This procedure takes as input an instance of the problem as well as a target makespan $C$ and either asserts correctly that there is no feasible schedule of makespan at most $C$, or returns a feasible schedule of makespan at most $\rho C$.
It is well-known that a $\rho$-relaxed decision procedure can be transformed into a polynomial-time $\rho$-approximation algorithm~\cite{HS85} provided that one can compute proper lower and upper bounds to the optimal value of size polynomial in the size of the input.

Given a target makespan $C$, let $$\gamma_j(C) := \min \{q \in \mathbb{Z}_{+} \suchthat f_j(q) \leq C \}$$ be the {\em critical speed} of job $j \in \J$. Moreover, we define for every $i \in \M$ the sets $J^+_i(C) := \{j \suchthat f(s_{i,j})\leq C\}$ and $J^-_i(C) := J \backslash J^+_i(C)$ to be the set of jobs that can or cannot be processed by $i$ alone within time $C$, respectively. 
Note that $\gamma_j(C)$ can be computed in polynomial-time by performing binary search when given oracle access to $f_j$.
When $C$ is clear from the context, we use the short-hand notation $\gamma_j$, $J^+_i$, and $J^-_i$ instead.
The following technical fact is equivalent to the non-decreasing work property and is used throughout the proofs of this paper:
\begin{fact}\label{fact:technical}
Let $f$ be a speed-implementable processing time function satisfying the properties of our problem. Then for every speed $q \in \mathbb{R}_{+}$ we have that:
\begin{enumerate}
\item $f(\alpha q) \leq \frac{1}{\alpha} f(q)$ for every $\alpha \in (0,1)$, and
\item $f(q') \leq \frac{q}{q'} f(q)$ for every $q' \leq q$.
\end{enumerate}
\end{fact}
\begin{proof}
For the first inequality, since $\alpha \in (0,1)$, it immediately follows that $\alpha q \leq q$. By the non-decreasing work property of $f$, we have that $\alpha q f(\alpha q) \leq q f(q)$, which implies that $f(\alpha q) \leq \frac{1}{\alpha} f(q)$. The second inequality is just an application of the first by setting $\alpha = \frac{q'}{q} \in (0,1)$.
\end{proof}
For every target makespan $C$, let $r_{i,j} := \max \{s_{i,j}, \gamma_j(C)\}$. The following feasibility LP is the starting point of the relaxed decision procedures we construct in this work:
\begin{ceqn}
\begin{align}
\text{[LP(C)]:}~ &\sum_{i \in \M} x_{i,j}  = 1~, \forall j \in \J \label{lp-assign}\\ 
&\sum_{j \in \J} \frac{f_j(r_{i,j})r_{i,j}}{s_{i,j}} x_{i,j} \leq C ~, \forall i \in \M \label{lp-makespan}\\
&x_{i,j} \geq 0 ~,\forall j \in \J, i \in \M. \label{lp-posit}
\end{align}
\end{ceqn}

In the above LP, each variable $x_{i,j}$ can be thought of as the fraction of job $j$ that is assigned to machine $i$. The equality constraints \eqref{lp-assign} ensure that each job is fully assigned to a subset of machines, while constraints \eqref{lp-makespan} impose an upper bound to the load of every machine.
Notice that for any job $j$ and machine $i$ such that $j \in J^+_i(C)$, it has to be that $s_{i,j} \geq \gamma_{j}$ and, thus, for the corresponding coefficient of constraints \eqref{lp-makespan}, we have $\frac{f_j(r_{i,j}) r_{i,j}}{s_{i,j}} = f_j(s_{i,j})$. Similarly, for any machine $i$ and job $j$ such that $j \in J^-_i$, we have that $s_{i,j} \leq \gamma_{j}$ and, thus, $\frac{f_j(r_{i,j}) r_{i,j}}{s_{i,j}} = \frac{f_j(\gamma_{j}) \gamma_{j}}{s_{i,j}}$.

As we prove in the following proposition, the above formulation is feasible for any $C$ that is greater than the optimal makespan.
\begin{proposition}\label{speed:lowerbound}
For every $C \geq \OPT$, where $\OPT$ is the makespan of an optimal schedule, [LP(C)] has a feasible solution.
\end{proposition}
\begin{proof}
Fix any feasible schedule of makespan $\OPT$ and let $S_j \subseteq \M$ be the set of machines allocated to a job $j$ in that schedule. 
For every $i \in \M$,$j \in \J$ set $x_{i,j} = \frac{s_{i,j}}{\sigma_j(S_j)}$ if $i \in S_j$ and $x_{i,j} = 0$, otherwise. We show that $x$ is a feasible solution to $[LP(C)]$.
Indeed, constraints \eqref{lp-assign} are satisfied since $\sum_{i \in \M} x_{i,j} = \sum_{i \in S_j} \frac{s_{i,j}}{\sigma_j(S_j)} = 1$ for all $j \in \J$. 
For verifying that constraints \eqref{lp-makespan} are satisfied, let $j \in \J$ and $i \in S_j$. By definition of $S_j$, it has to be that $\sigma_j(S_j) \geq s_{i,j}$ and $\sigma_j(S_j) \geq \gamma_{j}$, which implies that $\sigma_j(S_j) \geq r_{i,j}$. Therefore, by replacing $x_{i,j} = \frac{s_{i,j}}{\sigma_j(S_j)}$, the corresponding coefficient of \eqref{lp-makespan} becomes: $\frac{f_j(r_{i,j}) r_{i,j}}{s_{i,j}} \frac{s_{i,j}}{\sigma_j(S_j)} = \frac{f_j(r_{i,j}) r_{i,j}}{\sigma_j(S_j)} \leq f_j(S_j)$, where the last inequality follows by Fact \ref{fact:technical} and the fact that $\sigma_j(S_j) \geq r_{i,j}$. Using the above analysis, we can see that for any $i \in \M$ we have
$$\sum_{j \in \J} \frac{f_j(r_{i,j}) r_{i,j}}{s_{i,j}} x_{i,j} \leq \sum_{j \in \J|~ i \in S_j } f_j(S_j)\leq \OPT \leq C .$$ 
\end{proof}

Assuming that $C \geq \OPT$, let $x$ be an extreme point solution to [LP($C$)]. We create the {\em assignment graph} $\mathcal{G}(x)$ with nodes $V := \J \cup \M$ and edges $E := \{\{i,j\} \in \M \times \J \suchthat x_{i,j}>0\}$, i.e., one edge for each machine-job pair in the support of the LP solution. Notice that $\mathcal{G}(x)$ is bipartite by definition. 
Furthermore, since [LP($C$)] is structurally identical to the LP of unrelated machine scheduling~\cite{LST90}, the choice of $x$ as an extreme point guarantees the following sparsity property:

\begin{proposition}\cite{LST90}
For every extreme point solution $x$ of [LP(C)], each connected component of $\mathcal{G}(x)$ contains at most one cycle.
\end{proposition}
As a graph with at most one cycle is either a tree or a tree plus one edge, the connected components of $\mathcal{G}(x)$ are called \emph{pseudotrees} and the whole graph is called a {\em pseudoforest}.
It is not hard to see that the edges of an undirected pseudoforest can always be oriented in a way that every node has an {\em in-degree} of at most one. We call such a $\mathcal{G}(x)$ a {\em properly oriented} pseudoforest. Such an orientation can easily be obtained for each connected component as follows: We first orient the edges on the unique cycle (if it exists) consistently (e.g., clockwise) so as to obtain a directed cycle. 
Then, for every node of the cycle that is also connected with nodes outside the cycle, we define a subtree using this node as a root and we direct its edges away from that root (see e.g., Figure~\ref{fig:pseudo}).

Now fix a properly oriented $\mathcal{G}(x)$ with set of oriented edges $\bar{E}$. For $j \in \J$, we define $p(j) \in \M$ to be its unique {\em parent-machine} with $(p(j), j) \in \bar{E}$, if it exists, and $T(j) = \{i \in \M \suchthat (j,i) \in \bar{E}\}$ to be the set of {\em children-machines} of $j$, respectively. Notice, that for every machine $i$, there exists at most one $j \in \J$ such that $i \in T(j)$. The decision procedures we construct in this paper are based, unless otherwise stated, on the following scheme:

\textsc{Algorithm}: Given a target makespan $C$: 
\begin{enumerate}
\item If [LP($C$)] is feasible, compute an extreme point solution $x$ of [LP($C$)] and construct a properly oriented $\mathcal{G}(x)$. (Otherwise, report that $C < \OPT$.)
\item A {\em rounding scheme} assigns every job $j \in \J$ either only to its parent machine $p(j)$, or to a subset of its children-machines $T(j)$ (see Section \ref{sec:rounding}).
\item According to the rounding, every job $j \in \J$ that has been assigned to $T(j)$ is placed at the beginning of the schedule (these jobs are assigned to disjoint sets of machines).
\item At any point a machine $i$ becomes idle, it processes any unscheduled job $j$ that has been rounded to $i$ such that $i = p(j)$.
\end{enumerate}

\begin{figure}[t]
\centering
\begin{minipage}[b]{0.4\textwidth}
\centering
\begin{tikzpicture}[->,>=stealth',shorten >=0.5pt,auto,node distance=1.5cm,
                    thick,job node/.style={circle,draw,font=\sffamily\large\bfseries}, machine node/.style={rectangle,draw,font=\sffamily\large\bfseries}]

  \node[machine node] (1) {$i_1$};
  \node[job node] (2) [below left of=1] {$j_1$};
  \node[machine node] (3) [above left of=2]  {$i_2$};
  \node[job node] (4) [above right of=3] {$j_2$};
  \node[job node] (5) [right of=1] {$j_3$};
  \node[machine node] (6) [below left of=5]  {$i_3$};
  \node[machine node] (7) [below right of=5]  {$i_4$};
  \node[machine node] (8) [above of=1]  {$i_5$};

  \path[every node/.style={font=\sffamily\small}]
    (1) edge  node[left] {} (2)
    edge  node[left] {} (5)
    (2) edge node [right] {} (3)
    (3) edge node[right] {} (4)
    (4) edge node[right] {} (1)
        edge node[right] {} (8)
    (5) edge node[left] {} (6)
    edge node[left] {} (7);
    \draw[dashed, lightgray] (0,0) ellipse (0.6cm and 0.6cm);
    \draw[dashed, lightgray] (1.5,-1.1) ellipse (1.5cm and 0.7cm);
    \node[scale = 1, text width=1cm, red] at (1.5,-1.5) {$T(j_3)$};
	 \node[scale = 1, text width=1cm, red] at (0.75,0.7) {$p(j_3)$};
\end{tikzpicture}\vspace*{-0.55cm}
\caption{A properly oriented pseudotree with indegree at most $1$ for each node.}
\label{fig:pseudo}
\end{minipage}
\hspace*{2em}
\begin{minipage}[b]{0.4\textwidth}
\centering
\resizebox{0.8\textwidth}{!}{%
\begin{tikzpicture}[label/.style={ postaction={ decorate,transform shape, decoration={ markings, mark=at position .5 with \node #1;}}}]
      \draw[->,thick] (0,0) -- (5,0) node[right] {$\theta$};
      \draw[->,thick] (0,0) -- (0,5) node[right] {$g(\theta) = \sigma_j(S_j(\theta))$};
      \draw[scale=1,domain=0:5,thick, smooth,variable=\y,red, label={[above]{$\frac{\gamma_j f_j(\gamma_j)}{(\alpha+2\theta-2)C}$}}]  plot ({\y},{5/(3.163 + 2*\y/5 -2)}) ;
      \draw[color=black, thick, name path=A] plot coordinates {
		(0.02,3.5)
		(0.6,3.5)
		(0.6,3)
		(1.2,3)
		(1.2,2.6)
		(1.7,2.6)
		(1.7,2.0)
		(2.7,2.0)
		(2.7,1.8)
		(3.7,1.8)
		(3.7,1.5)
		(4.7,1.5)
		(4.7,1.4)
		(4.8,1.4)
		(4.8,0)
	};
	\draw[scale=1,domain=0.02:4.8,smooth,variable=\y,name path=B]  plot ({\y},{0.02}) ;
	\tikzfillbetween[of=A and B] {color=gray!20};
	\node[scale = 0.7, text width=5cm] at (2,1) {$$\int_{0}^{1} g(\theta) d\theta = \sum_{i \in T(j)} s_{i,j} (1 - \frac{\ell_i}{C})$$};
	\node[scale = 1, text width=1cm] at (0.4,-0.2) {$0$};
	\node[scale = 1, text width=1cm] at (5.2,-0.2) {$1$};
\end{tikzpicture}
}\vspace*{-0.55cm}
\caption{Volume argument for selecting a subset of the children machines in the proof of Proposition~\ref{speed:filtering}.}
\label{fig:fn}
\end{minipage}
\end{figure}

\section{Rounding schemes} \label{sec:rounding}
In each of the following rounding schemes, we are given as an input an extreme point solution $x$ of [LP(C)] and a properly oriented pseudoforest $\mathcal{G}(x)=(V, \bar{E})$.

\subsection{A simple $4$-approximation for unrelated machines}
\label{sec:rounding:simple}

We start from the following simple rounding scheme: For each job $j$, assign $j$ to its parent-machine $p(j)$ if $x_{p(j),j} \geq \frac{1}{2}$, or else, assign $j$ to its children-machines $T(j)$.
Formally, let $\J^{(1)} := \{j \in \J \suchthat x_{p(j),j} \geq \frac{1}{2}\}$ be the sets of jobs that are assigned to their parent-machines and $\J^{(2)} := \J \setminus \J^{(1)}$ the rest of the jobs. Recall that we first run the jobs in $\J^{(2)}$ and then the jobs in $\J^{(1)}$ as described at the end of the previous section.
For $i \in \M$, define $J^{(1)}_i := \{j \in \J^{(1)} \suchthat p(j) = i\}$ and $J^{(2)}_i := \{j \in \J^{(2)} \suchthat i \in T(j)\}$ as the sets of jobs in $\J^{(1)}$ and $\J^{(2)}$, respectively, that get assigned to $i$ (note that $|J^{(2)}_i| \leq 1$, as each machine gets assigned at most one job as a child-machine).
Furthermore, let $\ell_i := \sum_{j \in J_i^{(1)}} f_j(r_{i,j}) \frac{r_{i,j}}{s_{i,j}}x_{i,j}$ be the fractional load incurred by jobs in $\J^{(1)}$ on machine $i$ in the LP solution $x$.

\begin{proposition}\label{speed:parent}
Let $i \in \M$. Then $\sum_{j \in J^{(1)}_i} f_j(\{i\}) \leq 2\ell_i$.
\end{proposition}
\begin{proof}
Let $j \in J^{(1)}_i$.
Since $x_{i,j}\geq \frac{1}{2}$ by definition of $\J^{(1)}$, we get $f_j(s_{i,j}) \leq 2 f_j(s_{i,j}) x_{i,j}$.
Furthermore, since $s_{i,j} \leq \max\{s_{i,j}, \gamma_j\} = r_{i,j}$, the by Fact~\ref{fact:technical}, we have that $f_j(\{i\}) = f_j(s_{i,j}) \leq f_j(r_{i,j}) \frac{r_{i,j}}{s_{i,j}}$. Thus, by summing up over all jobs in $J^{(1)}_i$, we get
\begin{ceqn}
$$\sum_{j \in J^{(1)}_i} f_j(\{i\}) \leq 2 \sum_{j \in J^{(1)}_i} f_j(r_{i,j}) \frac{r_{i,j}}{s_{i,j}} x_{i,j} = 2\ell_i.$$
\end{ceqn}
\end{proof}

\begin{proposition}\label{speed:children}
Let $j \in \J^{(2)}$. Then $f_j(T(j)) \leq 2C$.
\end{proposition}
\begin{proof}
We assume that for all $i \in T(j)$, $s_{i,j} \leq \gamma_j$, since, otherwise, given some $i' \in T(j)$ with $s_{i,j} > \gamma_j$, we trivially have that $f_j(T(j)) \leq f_j(\{i'\}) \leq f_j(\gamma_j) \leq C$. 

For any $i \in T(j)$ and using the fact that $r_{i,j} = \gamma_j$, constraints \eqref{lp-makespan} imply that $f_j(\gamma_j) \frac{\gamma_j}{s_{i,j}} x_{i,j} \leq C$ for all $i \in T(j)$. 
Summing over all these constraints yields $\sum_{i \in T(j)} \frac{f_j(\gamma_j)}{C}\gamma_j x_{i,j} \leq \sigma_i(T(j))$.
Using the fact that $\sum_{i \in T(j)} x_{i,j} > \frac{1}{2}$ because $j \in  \J^{(2)}$, we get $\sigma_i(T(j)) \geq \frac{1}{2}\gamma_j \frac{f_j(\gamma_j)}{C}$. By combining this with the fact that $f_j(\gamma_j) \leq C$, by definition of $\gamma_j$, and using Fact~\ref{fact:technical}, this implies that $f_j(T(j)) \leq 2C$.

\end{proof}

Clearly, the load of any machine $i \in \M$ in the final schedule is the sum of the load due to the execution of $\J^{(1)}$, plus the processing time of at most one job of $\J^{(2)}$. By Propositions \ref{speed:parent}, \ref{speed:children} and the fact that $\ell_i \leq C$ for all $i \in \M$ by constraints \eqref{lp-makespan}, it follows that any feasible solution of [LP(C)] can be rounded in polynomial-time into a feasible schedule of makespan at most $4C$.

\subsection{An improved $\frac{2e}{e-1} \approx 3.163$-approximation for unrelated machines}
\label{sec:improved-rounding}

In the simple rounding scheme described above, it can be the case that the overall makespan improves by assigning some job $j \in \J^{(2)}$ only to a subset of the machines in $T(j)$. This happens because some machines in $T(j)$ may have significantly higher load from jobs of $\J^{(1)}$ than others, but job $j$ will incur the same additional load to all machines it is assigned to.

We can improve the approximation guarantee of the rounding scheme by taking this effect into account and filtering out children-machines with a high load. Define $\J^{(1)}$ and $\J^{(2)}$ as before. Every job in $j \in J^{(1)}$ is assigned to its parent-machine $p(j)$, while every job $j \in \J^{(2)}$ is assigned to a subset of $T(j)$, as described below.

For $j \in \J^{(2)}$ and $\theta \in [0, 1]$ define $S_j(\theta) := \{ i \in T(j) \suchthat 1 - \frac{\ell_i}{C} \geq \theta\}$. Choose $\theta_j$ so as to minimize $2(1 - \theta_j)C + f_j(S_j(\theta_j))$ (note that this minimizer can be determined by trying out at most $|T(j)|$ different values for $\theta_j$). 
We then assign each job in $j \in J^{(2)}$ to the machine set $S_j(\theta_j)$.


By Proposition \ref{speed:parent}, we know that the total load of each machine $i \in \M$ due to the execution of jobs from $\J^{(1)}$ is at most $2\ell_i$. Recall that there is at most one $j \in J^{(2)}$ with $i \in T(j)$. If $i \notin S_j(\theta_j)$, then load of machine $i$ bounded by $2\ell_i \leq 2C$. If $i \in S_j(\theta_j)$, then the load of machine $i$ is bounded by
\begin{ceqn} 
\begin{align}\label{speed:load}
\max_{i' \in S_j(\theta_j)}\Big\{2\ell_{i} + f_j(S_j(\theta_j)) \Big\} \leq 2(1 - \theta_j)C + f_j(S_j(\theta_j)),
\end{align}
\end{ceqn}
where the inequality comes from the fact that $1 - \frac{\ell_{i'}}{C} \geq \theta_j$ for all $i' \in S_{\theta_j}$. The following proposition gives an upper bound on the RHS of \eqref{speed:load} as a result of our filtering technique and proves Theorem~\ref{thm:filtering-approx}.
\begin{proposition}\label{speed:filtering}
For each $j \in \J^{(2)}$, there exists a $\theta \in [0,1]$ such that $2(1- \theta)C + f_j(S_j(\theta)) \leq \frac{2e}{e-1} C$.
\end{proposition}
\begin{proof}
We first assume that for all $i \in T(j)$, it is the case that $s_{i,j} \leq \gamma_j$ and, thus, $r_{i,j} = \gamma_j$. In the opposite case, where there exists some $i' \in T(j)$ such that $s_{i,j} > \gamma_j$, by choosing $\theta = 0$, then \eqref{speed:load} can be upper bounded by $2C + f_j(S_j(0)) = 2C + f_j(T(j)) \leq 2C + f_j(s_{i',j}) \leq 3C$ and the proposition follows.

Define $\alpha := \frac{2e}{e-1}$. We show that there is a $\theta \in [0, 1]$ with $\sigma_j(S_j(\theta)) \geq \frac{\gamma_j f_j(\gamma_j)}{(\alpha + 2 \theta -2)C}$. Notice that, in that case, $f_j(S_j(\theta)) \leq (\alpha + 2 \theta -2) C$ by Fact~\ref{fact:technical}, implying the lemma.


Define the function $g: [0,1] \rightarrow \mathbb{R}_{+}$ by $g(\theta) := \sigma_j (S_j(\theta))$. It is easy to see $g$ is non-increasing integrable and that 
 $$\int_{0}^{1} g(\theta) d\theta = \sum_{i \in T(j)} s_{i,j} (1 - \frac{\ell_i}{C}).$$ 
 See Figure \ref{fig:fn} for an illustration.

Now assume by contradiction that $g(\theta) < \frac{\gamma_j f_j(\gamma_j)}{(\alpha + 2 \theta -2)C}$ for all $\theta \in [0, 1]$.
Note that $\ell_i + \frac{\gamma_j f_j(\gamma_j)}{s_{i,j}} x_{i,j} \leq C$ for every $i \in T(j)$ by constraints \eqref{lp-makespan} and the fact that $|J^{(2)}_i| \leq 1$.
Hence  $\frac{f_j(\gamma_j)\gamma_j}{C} x_{i,j} \leq s_{i,j}(1 - \frac{\ell_i}{C})$ for all $i \in T(j)$.
Summing over all $i \in T(j)$ and using the fact that $\sum_{i \in T(j)} x_{i,j} \geq \frac{1}{2}$ because $j \in \J^{(2)}$, we get
\begin{ceqn}
\begin{align*}
\frac{f_j(\gamma_j)\gamma_j}{2 C} &\leq \sum_{i \in T(j)} s_{i,j} (1 - \frac{\ell_i}{C}) \\
&= \int_{0}^{1} g(\theta) d\theta \\
&< \frac{f_j(\gamma_j)\gamma_j}{C} \int_{0}^{1} \frac{1}{\alpha + 2 \theta -2} d\theta,
\end{align*}
\end{ceqn}
where the last inequality uses the assumption that $g(\theta) < \frac{\gamma_j f_j(\gamma_j)}{(\alpha + 2 \theta_j -2)C}$ for all $\theta \in [0,1]$. By simplifying the above inequality we get the contradiction
\begin{ceqn}
\[1 < \int_{\alpha-2}^{\alpha} \frac{1}{\lambda} d\lambda = \ln(\frac{\alpha}{\alpha -2}) = 1,\]
\end{ceqn}
which concludes the proof.
\end{proof}
By the above analysis, our main result for the case of unrelated machines follows.

\restateUnrelated*

\begin{remark}We can slightly improve the above algorithm by optimizing over the threshold of assigning each job to the parent- or children-machines in the assignment graph (see Appendix~\ref{appendix:unrelated:tuning} for details). This optimization gives a slightly better approximation guarantee of $\alpha = \inf_{\beta \in (0,1)} \left\{ \frac{e^{\frac{1}{\beta} - 1}}{\beta(e^{\frac{1}{\beta} - 1} - 1)} \right\} \approx 3.14619$.
\end{remark}

\begin{remark}
The LP-based nature of our techniques allows the design of polynomial-time $\mathcal{O}(1)$-approximation
algorithms for the objective of minimizing the sum of weighted completion times, i.e., $\sum_{j \in \J} w_j C_j$, where $w_j$ is the {\em weight} and $C_j$ is the {\em completion time} of job $j \in \J$. This can be achieved by using the rounding theorems of this section in combination with the standard technique of interval-indexed formulations \cite{HSSW97}.
\end{remark}

\subsection{A $7/3$-approximation for restricted identical machines}
\label{sec:restricted-rounding}

We are able to provide an algorithm of improved approximation guarantee for the special case of restricted identical machines. In this case, each job $j\in \J$ is associated with a set of machines $\M_j \subseteq \M$, such that $s_{i,j} = 1$ for $i \in \M_j$ and $s_{i,j} = 0$, otherwise. Notice that our assumption that $f_j(0) = + \infty, \forall j \in \J$ implies that, in the restricted identical machines case, every job $j$ has to be scheduled only on the machines of $\M_j$.

Given a feasible solution to [LP(C)] and a properly oriented $\mathcal{G}(x)$, we define the sets $\J^{(1)} := \{j \in \J \suchthat x_{p(j),j} = 1\}$ and $\J^{(2)} := \J \setminus \J^{(1)}$. The rounding scheme for this special case can be described as follows:
\begin{enumerate}
    \item Every job $j \in \J^{(1)}$ is assigned to $p(j)$ (which is the only machine in $\mathcal{G}(x)$ that is assigned to $j$).
    \item Every job $j \in \J^{(2)}$
    \begin{enumerate}
        \item is assigned to the set $T(j)$ of its children-machines, if $|T(j)| = 1$ or $|T(j)| \geq 3$.
        \item is assigned to the subset $S \subseteq T(j)$ that results in the minimum makespan over $T(j)$, if $|T(j)| = 2$. Notice that for $|T(j)|=2$ there are exactly three such subsets.
    \end{enumerate}
    \item As usual, the jobs of $\J^{(2)}$ are placed at the beginning of the schedule, followed by the jobs of $\J^{(1)}$.
\end{enumerate}

Clearly, by definition of our algorithm, constraints \eqref{lp-makespan} and the fact that $f_j(1) \leq f_j(\gamma_j) \gamma_j$ for all $j \in \J$, the load of any machine that only processes jobs of $\J^{(1)}$ is at most $C$. Therefore, we focus our analysis on the case of machines that process jobs from $\J^{(2)}$. Recall, that every machine $i \in \M$ can process at most one job $j \in \J^{(2)}$ and, thus, the rest of the proof is based on analyzing the makespan of the machines of $T(j)$, for each $j \in \J^{(2)}$.

\begin{proposition} \label{restricted-basic}
Any job $j \in \J^{(2)}$ can be assigned to the set $T(j)$ with processing time most $\frac{|T(j)| + 1}{|T(j)|} C$.
\end{proposition}
\begin{proof}
Fix any job $j \in \J^{(2)}$. By summing over the constraints \eqref{lp-makespan} for $i \in T(j) \cup \{p(j)\}$ (every machine in the support of $j$) and using constraints \eqref{lp-assign}, we have that $\gamma_j f_j(\gamma_j) \leq (|T_j|+1) C$. By applying Fact \ref{fact:technical} and the non-increasing property of $f_j$, we have that: 
\begin{ceqn}
\begin{align*}
   f_j(|T(j)|) &\leq \frac{|T(j)| + 1}{|T(j)|} f_j(|T(j)| + 1) \\
   &\leq \frac{|T(j)| + 1}{|T(j)|} f_j(\frac{f_j(\gamma_j)}{C} \gamma_j) \\
   &\leq \frac{|T(j)| + 1}{|T(j)|} \frac{C}{f_j(\gamma_j)} f_j(\gamma_j) \\
   &\leq \frac{|T(j)| + 1}{|T(j)|} C. 
\end{align*}
\end{ceqn}
\end{proof}
Taking into account that for any machine $i \in \M$, the load due to the jobs of $\J^{(1)}$ is at most $C$, the above proposition gives a makespan of at most $\left(1 + \frac{4}{3}\right)C$, for the machines $T(j)$ of every job $j \in \J^{(2)}$ such that $|T(j)| \geq 3$. The cases where $|T(j)| \in \{1,2\}$ need a more delicate treatment. For any $i \in T(j)$, let $\ell_i := \sum_{j' \in \J^{(1)} | p(j') = i} f_{j'}(1) \leq C$ to be the load of $i$ w.r.t. the jobs of $\J^{(1)}$.
\begin{proposition}
For any job $j \in \J^{(2)}$ with $|T(j)| = 1$ that is assigned to its unique child $i \in T(j)$, the total load of $i$ is at most $2 C$.
\end{proposition}
\begin{proof}
Consider a job $j \in \J^{(2)}$ of critical speed $\gamma_j$, such that $|T(j)| = 1$, and let $i$ be the unique child-machine in $T(j)$. By our assumption that $i \in \M_j$ and the fact that $\gamma_j \geq 1$, we have that $r_{i,j} = \gamma_{i,j}$. By constraints \eqref{lp-makespan} of [LP(C)], it is the case that $\ell_i + \gamma_j f_j(\gamma_j) x_{i,j} \leq C$. Therefore, since our algorithm assigns job $j$ to $i$, the load of the latter becomes: $\ell_i + f_j(1) \leq \ell_i + \gamma_j f_j(\gamma_j)x_{i,j} + \gamma_j f_j(\gamma_j)(1-x_{i,j}) \leq C + \gamma_j f_j(\gamma_j)x_{p(j),j}$, where we used the fact that $f_j(1) \leq \gamma_j f_j(\gamma_j)$ and that $x_{i,j} + x_{p(j),j} = 1$ by constraints \eqref{lp-assign}. However, by constraints \eqref{lp-makespan} for $p(j) \in \M$, we can see that $\gamma_j f_j(\gamma_j)x_{p(j),j} \leq C$, which completes the proof.
\end{proof}

\begin{proposition}
Consider any job $j \in \J^{(2)}$ with $|T(j)| = 2$ that is assigned to the subset $S \subseteq T(j)$ that results in the minimum load. The load of any machine $i \in S$ is at most $\frac{9}{4}C$.
\end{proposition}
\begin{proof}
We first notice that for the case where $\gamma_j \leq |T(j)| = 2$, by assigning $j$ to all the machines of $T(j)$, the total load of any $i \in T(j)$ is at most: $\ell_i + f_j(T(j)) \leq \ell_i + f_j(\gamma_j) \leq 2C$, by constraints \eqref{lp-makespan}. Therefore, we focus on the case where $\gamma_j \geq 3$ (since we assume that $\gamma_j$ is an integer). For the machines of $T(j)$, denoted by $T(j) = \{i_1,i_2 \}$, we assume w.l.o.g. that $x_{i_1,j} \geq x_{i_2,j}$. 

Our algorithm attempts to schedule $j$ on the sets $S_1 = \{i_1\}$, $S_2 = \{i_2\}$ and $S_B = \{i_1, i_2\}$ and returns the assignment of minimum makespan. We can express the maximum load of $T(j)$ in the resulting schedule as:
\begin{ceqn} 
\begin{align}
    &\min \Big\{ f_j(1) + \ell_{i_1}, f_j(1) + \ell_{i_2}, f_j(2) + \max \big\{ \ell_{i_1}, \ell_{i_2} \big\}  \Big\}\nonumber\\
    &\leq C+\min \Big\{ f_j(1) - \gamma_j f_j(\gamma_j)x_{i_1,j}, f_j(1) - \gamma_j f_j(\gamma_j)x_{i_2,j}, f_j(2) +\gamma_j f_j(\gamma_j) \max \big\{- x_{i_1,j}, -x_{i_2,j} \big\}  \Big\}\nonumber\\
    &\leq  C + \min \Big\{f_j(1) - \gamma_j f_j(\gamma_j) x_{i_1,j}, f_j(2) - \gamma_j f_j(\gamma_j) x_{i_2,j} \Big\}\nonumber\\
    &\leq C + \frac{1}{2} \Big( f_j(1) +  f_j(2) - \gamma_j f_j(\gamma_j) (x_{i_1,j} + x_{i_2,j}) \Big) \nonumber\\ 
    &= C + \frac{1}{2} \Big( f_j(1) +  f_j(2) - \gamma_j f_j(\gamma_j) (1 - x_{p(j),i}) \Big) \label{restricted:temp},
\end{align}
\end{ceqn} 
where the first inequality follows by the fact that $\ell_i \leq C - \gamma_j f_j(\gamma_j)x_{i,j}$, by constraints \eqref{lp-makespan}. Furthermore, the second inequality follows by the assumption that $x_{i_1,j} \geq x_{i_2,j}$, while the third inequality follows by balancing the two terms of the minimization. Finally the equality follows by the fact that $x_{i_1,j} + x_{i_2,j} = 1 - x_{p(j),j}$, by constraints \eqref{lp-assign}.

By constraints \eqref{lp-makespan}, we get that $x_{p(j),j} \leq \frac{C}{\gamma_j f_j(\gamma_j)}$, while by Fact \ref{fact:technical} we have that: $\gamma_j f_j(\gamma_j) \geq f_j(1)$, given that $\gamma_j \geq 3$. Moreover, by Fact \ref{fact:technical}, we have that $f_j(2) \leq \frac{3}{2} C$, using the analysis of Proposition \ref{restricted-basic}. By combining the above, we get that 
\begin{ceqn} 
\begin{align*}
    \eqref{restricted:temp}&\leq C + \frac{1}{2} \Big( f_j(1) +  f_j(2)  - \gamma_j f_j(\gamma_j) +  \gamma_j f_j(\gamma_j) x_{p(j),i}) \Big) \\
    &\leq C + \frac{1}{2} \Big(  f_j(1) +  f_j(2) -  f_j(1) + C \Big) \\
    &\leq \frac{9}{4} C.
\end{align*}
\end{ceqn} 
\end{proof}
By considering the worst of the above scenarios, we can verify that the makespan of the produced schedule is at most $\frac{7}{3} C$, thus, leading to the following theorem.

\restateRestricted*

\subsection{A $3$-approximation for uniform machines}
\label{sec:uniform-rounding}

We prove an algorithm of improved approximation guarantee for the special case of uniform machines, namely, every machine $i \in \M$ is associated with a unique speed $s_i$, such that $s_{i,j} = s_i$ for all $j \in \J$. Given a target makespan $C$, we say that a machine $i$ is {\em $j$-fast} for a job $j \in \J$ if $f_j(\{i\}) \leq C$, while we say that $i$ is {\em $j$-slow} otherwise. 
As opposed to the previous cases, the rounding for the uniform case starts by transforming the feasible solution of [LP($C$)] into another extreme point solution that satisfies a useful structural property, as described in the following proposition:

\begin{restatable}{proposition}{restatePropositionUniform}\label{uniform:bfs}
There exists an extreme point solution $x$ of [LP(C)] that satisfies the following property: For each $j \in \J$ there is at most one $j$-slow machine $i \in \M$ such that $x_{i,j} > 0$ and $x_{i,j'} > 0$ for some job $j' \neq j$. Furthermore, this machine, if it exists, is the slowest machine that $j$ is assigned to.
\end{restatable}
\begin{proof}
Consider a job $j$ and two $j$-slow machines $i_1, i_2 \in \M$, such that $x_{i_1,j}, x_{i_2,j} > 0$.
Let also two jobs $j_1, j_2 \in \J$, other than $j$, such that $x_{i_1, j_1}>0$ and $x_{i_2, j_2}>0$, assuming w.l.o.g. that $s_{i_2} \geq s_{i_1}$. We emphasize the fact $j_1$ and $j_2$ can correspond to the same job.
Recall that for any job $j'$ and machine $i'$, we have that $r_{i',j'} = \gamma_{j'}$, if $i'$ is $j'$-slow, and $r_{i',j'} = s_{i',j'}$, if $i'$ is $j'$-fast. 

We show that we can transform this solution into a new extreme point solution $x'$ such that one of the following is true: (a) $j$ is no longer supported by $i_1$ (i.e. $x'_{i_1,j} = 0$), or (b) $j_2$ is no longer supported by $i_2$ (i.e. $x'_{i_2,j_2} = 0$) and $x'_{i_1,j_2} = x_{i_1,j_2} +  x_{i_2,j_2}$. 
Let $a_{i',j'}$ be the coefficient of $x_{i',j'}$ in the LHS of \eqref{lp-makespan} for some machine $i'$ and job $j'$.
Since $i_1,i_2$ are $j$-slow machines, it is the case that 
$$a_{i_1,j} = f_j(\gamma_j)\frac{\gamma_j}{s_{i_1}} \geq f_j(\gamma_j)\frac{\gamma_j}{s_{i_2}} = a_{i_2,j}.$$
Suppose that we transfer an $\epsilon >0$ mass from $x_{i_1,j}$ to $x_{i_2,j}$ (without violating constraints \eqref{lp-assign}). Then the load of $i_1$ decreases by $\Delta_- = \epsilon f_j(\gamma_j)\frac{\gamma_j}{s_{i_1}}$, while the load of $i_2$ increases by $\Delta_+ = \epsilon f_j(\gamma_j)\frac{\gamma_j}{s_{i_2}}$. 
In order to avoid the violation of constraints \eqref{lp-makespan}, we also transfer an $\epsilon_2$ mass from $x_{i_2,j_2}$ to $x_{i_1,j_2}$, such that: 
$\epsilon_2 a_{i_2,j_2} = \Delta_+$, i.e. $\epsilon_2 = \epsilon \frac{1}{a_{i_2,j_2}}  f_j(\gamma_j)\frac{\gamma_j}{s_{i_2}}$, given that any coefficients $a_{i',j'}$ is strictly positive. Clearly, constraints \eqref{lp-makespan} for $i_2$ are satisfied since the fractional load of the machine stays the same as in the initial feasible solution. 
It suffices to verify that constraints \eqref{lp-makespan} for $i_1$ are also satisfied. 
Indeed, the difference in the load of $i_1$ that results from the above transformation can be expressed as $\epsilon_2 a_{i_1,j_2} - \Delta_-$. 
However, it is the case that: 
\begin{ceqn}
\begin{align*}
    \epsilon_2 a_{i_1,j_2} &= \epsilon \frac{a_{i_1,j_2}}{a_{i_2,j_2}}  f_j(\gamma_j)\frac{\gamma_j}{s_{i_2}} \\
    &\leq \epsilon \frac{s_{i_2}}{s_{i_1}}  f_j(\gamma_j)\frac{\gamma_j}{s_{i_2}} = \epsilon f_j(\gamma_j)\frac{\gamma_j}{s_{i_1}} = \Delta_-
\end{align*}
\end{ceqn}
and therefore, constraint \eqref{lp-makespan} of $i_1$ remains feasible as the load difference is non-positive. 

In the last inequality, we used the fact that $\frac{a_{i_1,j_2}}{a_{i_2,j_2}} \leq \frac{s_{i_2}}{s_{i_1}}$, which can be proved by case analysis: 

(i) If both $i_1,i_2$ are $j_2$-slow, then clearly $$\frac{a_{i_1,j_2}}{a_{i_2,j_2}} = \frac{f_j(\gamma_{j_2}) \gamma_{j_2} / s_{i_1} }{f_j(\gamma_{j_2}) \gamma_{j_2} /s_{i_2}}  =  \frac{s_{i_2}}{s_{i_1}}.$$

(ii) if $i_2$ is $j_2$-fast and $i_1$ is $j_2$-slow, we have that: $$\frac{a_{i_1,j_2}}{a_{i_2,j_2}} = \frac{f_j(\gamma_{j_2}) \gamma_{j_2}/s_{i_1}}{f_{j_2}(s_{i_2})} \leq \frac{s_{i_2}}{s_{i_1}},$$ since by Fact \ref{fact:technical}, we have that $f_{j_2}(s_{i_2}) \geq f_j(\gamma_{j_2}) \frac{\gamma_{j_2}}{s_{i_2}}$. 

(iii) If both $i_1,i_2$ are $j_2$-fast, then $$\frac{a_{i_1,j_2}}{a_{i_2,j_2}} = \frac{f_{j_2}(s_{i_1})}{f_{j_2}(s_{i_2})} \leq \frac{s_{i_2}}{s_{i_1}},$$ which follows by Fact \ref{fact:technical} and the fact that $s_{i_2} \geq s_{i_1}$.

By the above analysis, we can keep exchanging mass in the aforementioned way until either $x_{i_1,j}$ or $x_{i_2, j_2}$ becomes zero. In any case, job $j$ shares at most one of $i_1$ and $i_2$ with another job, the above process. 
Notice that the above transformation always returns a basic feasible solution and reduces the total number of shared $j$-slow machines for a job $j$ by one. Therefore, by applying the transformation at most $|\J|+|\M|$ times, one can get a basic feasible solution that satisfies the required property. Finally, by the same argument it follows that for any job $j \in \J$ that shares a $j$-slow machine $i_j$ with other jobs, then $i_j$ is the slowest that $j$ is assigned to. In the opposite case, assuming that there is another $j$-slow machine $i'$ such that $s_{i_j} > s_{i'}$, repeated application of the above transformation would either move the total $x_{i',j}$ mass from $i'$ to $i_j$, or replace the assignment of all the shared jobs from $i_j$ to $i'$.
\end{proof}

Let $x$ be an extreme point solution of [LP(C)] that satisfies the property of Proposition~\ref{uniform:bfs} and let $\mathcal{G}(x)$ a properly oriented pseudoforest. By the above proposition, each job $j$ has at most three types of assignments in $\mathcal{G}(x)$, regarding its set $T(j)$ of children-machines: (i) $j$-fast machines $F_j \subseteq T(j)$, (ii) {\em exclusive} $j$-slow machines $D_j \subseteq T(j)$, i.e. $j$-slow machines that are completely assigned to $j$, and (iii) at most one {\em shared} $j$-slow machine $i_{j} \in T(j)$ (which is the slowest machine that $j$ is assigned to).

We now describe the rounding scheme for the special case of uniform machines.

\begin{enumerate}
    \item For any job $j \in \J$ such that $x_{p(j),j} \geq \frac{1}{2}$, $j$ is assigned to its parent-machine $p(j)$.
    \item For any job $j \in \J$ such that such that $x_{p(j),j} < \frac{1}{2}$, $j$ is assigned to a subset $S \subseteq T(j)$, according to the following rule:
        \begin{enumerate}
            \item If $F_j \neq \emptyset$, assign $j$ to any $i \in F_j$, \label{uniform:caseb1}
            \item else if $\sigma_j(D_j) \geq \frac{r_{i,j} f_j(r_{i,j})}{3 C}$, then assign $j$ only to the machines of $D_j$ (but not to the shared $i_{j}$). \label{uniform:caseb2}
            \item In any other case, $j$ is assigned to $D_j \cup \{i_j\}$. \label{uniform:caseb3}
        \end{enumerate}
    \item As usual, the jobs that are assigned to more than one machines are placed at the beginning of the schedule, followed by the rest of the jobs.
\end{enumerate}

Let $\J^{(1)}$ be the set of jobs that are assigned by the algorithm to their parent-machines and $\J^{(2)} = \J \setminus \J^{(1)}$ be the rest of the jobs. 
By using the same analysis as in Proposition \ref{speed:parent}, we can see that the total load of any machine $i \in \M$ that processes only jobs from $\J^{(1)}$ is at most $2C$.
Since any machine can process at most one job from $\J^{(2)}$, it suffices to focus on the makespan of the set $T(j)$ for any job $j \in \J^{(2)}$.

Clearly, for the case \ref{uniform:caseb1}, the processing time of a job $j$ on any machine $i \in F_j $ is at most $C$, which, in combination with Proposition \ref{speed:parent}, gives a total load of at most $3C$. 
Moreover, for the case \ref{uniform:caseb2}, where a job $j$ is scheduled on the machines of $D_j \subseteq T(j)$, the total load of every $i \in D_j $ is at most $3C$. This follows by a simple application of Fact \ref{fact:technical}, since $\sigma_j(D_j) \geq \frac{r_{i,j} f_j(r_{i,j})}{3 C} \geq \frac{\gamma_j f_j(\gamma_j)}{3 C}$, and the fact that the machines of $D_j$ are exclusively assigned to $j$. 

Finally, consider case \ref{uniform:caseb3} of the algorithm, where $\sigma_j(D_j) < \frac{r_{i,j} f_j(r_{i,j})}{3 C}$ and $j$ is scheduled on $(D_j \cup \{i_j\})$. By constraints \eqref{lp-makespan} and the fact that $x_{i_j, j} + \sum_{i \in D_j} x_{i,j} > \frac{1}{2}$, it follows that $\sigma_j(D_j\cup \{i_j\}) \geq f_j(r_{i,j}) \frac{r_{i,j}}{2 C}$, which by Fact \ref{fact:technical} implies that $f_j(D_j\cup \{i_j\}) \leq 2C$. Notice that the above analysis implies the existence of a shared $i_j$ machine, if the time the algorithm reaches case \ref{uniform:caseb3}.

In order to complete the proof, it suffices to show that the load of the machine $i_j$ (i.e. the only machine of $T(j)$ that is shared with other jobs of $\J^{(1)}$) is at most $C$. 
Let $\ell_i = \sum_{j' \in \J^{(1)}|i=p(j)} f(r_{i,j'}) \frac{r_{i,j'}}{s_{i,j'}} x_{i,j'}$ be the fractional load of a machine $i$ due to the jobs of $\J^{(1)}$. 
Clearly, if $i_j$ is the only child-machine of $j$, then it has to be that $x_{i_j,j}>\frac{1}{2}$ and the total load of $i_j$ after the rounding is at most $2 \ell_{i_j} + f_j(s_{i_j}) \leq 2 \ell_i + 2  f_j(s_{i_j}) x_{i_j,j} \leq 2 \ell_i + 2  f_j(r_{i,j}) \frac{r_{i,j}}{s_{i_j}} x_{i_j,j} \leq 2C$, using Proposition \ref{speed:parent} and Fact \ref{fact:technical}. 
Assuming that $D_j \neq \emptyset$, by summing over constraints \eqref{lp-makespan} for all $i \in D_j$ and using the fact that $\sigma_j(D_j) < \frac{r_{i,j} f_j(r_{i,j})}{3 C}$, we get that $\sum_{i \in D_j} x_{i,j} < \frac{1}{3}$. Therefore, combining this with the fact that $x_{i_j,j} + \sum_{i \in D_j} x_{i,j} > \frac{1}{2}$, we conclude that $x_{i_j,j} > \frac{1}{6}$. 
Now, for the fractional load $\ell_{i_j}$ of $i_j$ due to the jobs of $\J^{(1)}$, by constraints \eqref{lp-makespan}, we have that $\ell_i \leq C - f_j(r_{i,j}) \frac{r_{i,j}}{s_{i_j}} x_{i_j,j} \leq C - f_j(r_{i,j}) \frac{3 r_{i,j} C}{6 r_{i,j} f_j(r_{i,j})} \leq \frac{C}{2}$, where we used that $x_{i_j,j} > \frac{1}{6}$ and the fact that $i_j$ is the smallest machine in $T(j)$ and, thus, $s_{i_j} \leq \sigma_j (D_j) < \frac{r_{i,j} f_j(r_{i,j})}{3 C}$. By Proposition \ref{speed:parent}, the load of $i_j$ due to the jobs of $\J^{(1)}$ after the rounding is at most $2 \ell_{i_j} \leq C$. Therefore, in every case, the load of any machine $i \in T(j)$ for all $j \in \J^{(2)}$ is at most $3C$, which leads to the following theorem.

\restateUniform* 

\section{Extension: Sparse allocations via $p$-norm regularization}
\label{sec:effective-speed}
In the model of speed-implementable processing time functions, each function $f_j(S)$ depends on the total additive speed, yet is oblivious to the actual number of allocated machines. However, the overhead incurred by the synchronization of physical machines naturally depends on their number. In this section, we study an extension of the speed-implementable malleable model, that captures the impact of the cardinality of a set of machines through the notion of {\em effective speed}. In this setting, every job $j$ is associated with a speed {\em regularizer} $p_j \geq 1$, while the total speed of a set $S \subseteq \M$ is given by: $\sigma^{(p_j)}_j(S) = \left( \sum_{i \in S} s^{p_j}_{i,j} \right)^{\frac{1}{p_j}}$. For simplicity, we assume that every job has the same speed regularizer.

Clearly, the choice of $p$ controls the effect of the cardinality of a set to the resulting speed of an allocation, given that as $p$ increases a {\em sparse} (small cardinality) set has higher effective speed than a non-sparse set of the same total speed. Notice that for $p = 1$, we recover the standard case of additive speeds, while for $p \to \infty$, parallelization is no longer helpful as $\lim_{p \to \infty}\sigma^{(p)}_j(S) = \max_{i \in S}\{s_{i,j}\}$. As before, the processing time functions satisfy the standard properties of malleable scheduling, i.e., $f_j(s)$ is non-increasing while $f_j(s) \cdot s$ is non-decreasing in the total allocated speed.
For simplicity of presentation we assume that all jobs have the same regularizer~$p$, i.e., $p = p_j, \forall j \in \J$, but we comment on the case of job-dependent regularizers at the end of this section.

Quite surprisingly, we can easily modify the algorithms of the previous section in order to capture the above generalization. Given a target makespan $C$, we start from a new feasibility program [LP$^{(p)}$(C)], which is given by constraints \eqref{lp-assign}, \eqref{lp-posit} of [LP(C)], combined with: 

\begin{ceqn}
\begin{align}
\sum_{j \in \J} f_j(r_{i,j}) \left( \frac{r_{i,j}}{s_{i,j}} \right)^p x_{i,j} \leq C, \forall i \in \M. \label{plp-makespan}
\end{align}
\end{ceqn}
Note that $\gamma_j(C)$ and $r_{i,j} = \max\{\gamma_j(C),s_{i,j}\}$ are defined exactly as before. As we can see, the only difference between $[\text{LP}(C)]$ and $[\text{LP}^{(p)}(C)]$ is that we replace each coefficient $f_j(r_{i,j})\frac{r_{i,j}}{s_{i,j}}$ with $\left( f_j(r_{i,j}) \frac{r_{i,j}}{s_{i,j}} \right)^p$ in constraints \eqref{lp-makespan} of the former. As we show in the following proposition, for any $C \geq \OPT$, where $\OPT$ is the makespan of an optimal schedule, [LP$^{(p)}$(C)] has a feasible solution.

\begin{proposition}
For every $C \geq \OPT$, where $\OPT$ is the makespan of an optimal schedule, [LP$^{(p)}$(C)] has a feasible solution.
\end{proposition}
\begin{proof}
Consider an optimal solution of makespan $\OPT$, where every job $j$ is assigned to a subset of the available machines $S_j \subseteq \M$. Setting $x_{i,j} = \frac{s^p_{i,j}}{\sigma^p_j(S_j)}$ if $i \in S_j$ and $x_{i,j} = 0$, otherwise, we get a feasible solution that satisfies inequalities \eqref{lp-assign}, \eqref{plp-makespan}, \eqref{lp-posit}. Indeed, for constraints \eqref{lp-assign} we have $\sum_{i \in \M} x_{i,j} = \sum_{i \in S_j} \frac{s^p_{i,j}}{(\sigma^{(p)}_j(S_j))^p} = 1$. Moreover, for every $i \in S_j$ and $j \in \J$, it is true that:
\begin{ceqn}
\begin{align*}
    f_j(r_{i,j}) \left( \frac{r_{i,j}}{s_{i,j}} \right)^p x_{i,j} &= f_j(r_{i,j})  \left( \frac{r_{i,j}}{s_{i,j}} \right)^p \left(\frac{s_{i,j}}{\sigma^{(p)}_j(S_j)} \right)^p \\
    &\leq \left(\frac{r_{i,j}}{\sigma^{(p)}_j(S_j)} \right)^{p} f_j(r_{i,j}) \\
    &\leq \left(\frac{r_{i,j}}{\sigma^{(p)}_j(S_j)} \right)^{p-1} f_j(S_j) \\
    &\leq f_j(S_j),
\end{align*}
\end{ceqn}
where in the second inequality, we use Fact \ref{fact:technical}, since $r_{i,j} = \max\{s_{i,j}, \gamma_j(C)\} \leq \sigma^{(p)}_j(S_j)$. Moreover, we use the fact that $\frac{r_{i,j}}{\sigma^{(p)}_j(S_j)} \leq 1$ for all $j \in \J$ and $i \in S_j$. 
By the above analysis, we can see that constraints \eqref{plp-makespan} are satisfied, since for all $i \in \M$, we have: 
\begin{ceqn}
\begin{align*}
\sum_{j \in \J} f_j(r_{i,j}) \left( \frac{r_{i,j}}{s_{i,j}} \right)^p x_{i,j} \leq \sum_{j \in J|~i \in S_j } f_j(S_j) \leq \OPT \leq C.
\end{align*}
\end{ceqn}
\end{proof}

The algorithm for this setting is similar to the one of the standard case (see Section \ref{sec:rounding:simple}) and is based on rounding a feasible extreme point solution of $[LP^{(p)}(C)]$. Moreover, the rounding scheme is a parameterized version of the simple rounding of Section~\ref{sec:rounding:simple}, with the difference that the threshold parameter $\beta \in [0,1]$ (i.e., the parameter that controls the decision of assigning a job $j$ to either $p(j)$ or $T(j)$) is not necessarily $\frac{1}{2}$. In short, given a pseudoforest $\mathcal{G}(x)$ on the support of a feasible solution $x$, the rounding scheme assigns any job $j$ to $p(j)$ if $x_{p(j),j} \geq \beta$, or to $T(j)$, otherwise.

\begin{proposition}
Any feasible solution of  [LP$^{(p)}$(C)] can be rounded in polynomial-time into a feasible schedule of makespan at most $\left( \frac{1}{\beta} + \frac{1}{(1-\beta)^{1/p}} \right)C$.
\end{proposition}
\begin{proof}
In order to prove the upper bound on the makespan of the produced schedule, we work similarly to the proof of Section~\ref{sec:rounding:simple}.
Let $\J^{(1)}$ be the set of jobs that are assigned by our algorithm to their parent-machine and let $\J^{(2)} = \J \setminus \J^{(1)}$ be the rest of the jobs.

We first show that the total load of any machine $i \in \M$ incurred by the jobs of $\J^{(1)}$ is at most $\frac{1}{\beta}C$. By definition of our algorithm, for every job $j \in \J^{(1)}$ that is assigned to some machine $i$, it holds that $x_{i,j}\geq \beta$. For the load of any machine $i \in \M$ in the rounded schedule that corresponds to jobs of $\J^{(1)}$, we have: 
\begin{ceqn}
\begin{align*}
\sum_{j \in \J^{(1)}|i = p(j)} f_j(s_{i,j}) &= \frac{1}{\beta}\sum_{j \in \J^{(1)}|i = p(j)} f_j(s_{i,j}) \beta \\
&\leq \frac{1}{\beta}\sum_{j \in \J^{(1)}|i = p(j)} f_j(s_{i,j}) x_{i,j} \\
&\leq \frac{1}{\beta}\sum_{j \in \J^{(1)}|i = p(j)} f_j(r_{i,j}) \frac{r_{i,j}}{s_{i,j}} x_{i,j}\\
&\leq \frac{1}{\beta}\sum_{j \in \J^{(1)}|i = p(j)} f_j(r_{i,j}) \left(\frac{r_{i,j}}{s_{i,j}}\right)^p x_{i,j}\\
&\leq \frac{1}{\beta} C,
\end{align*}
\end{ceqn}
where the second inequality follows by Fact \ref{fact:technical}, given that $s_{i,j} \leq r_{i,j}$. Moreover, the third inequality follows by the fact that $\frac{r_{i,j}}{s_{i,j}} \geq 1$ and, thus, $\frac{r_{i,j}}{s_{i,j}} \leq \left(\frac{r_{i,j}}{s_{i,j}}\right)^p$. Finally, the last inequality follows by constraints \eqref{plp-makespan}.

The next step is to prove that every job $j \in \J^{(2)}$ has a processing time of at most $(\frac{1}{1-\beta})^{\frac{1}{p}} C$. By definition of the algorithm, for any $j\in \J^{(2)}$ it is the case that $\sum_{i \in T(j)} x_{i,j} > 1-\beta$. 
First notice that for any $j \in \J^{(2)}$ and $i \in T(j)$, by constraints \eqref{plp-makespan}, we have that $f_j(r_{i,j}) \left( \frac{r_{i,j}}{s_{i,j}} \right)^p x_{i,j} \leq C$, which is equivalent to $f_j(r_{i,j})  r_{i,j}^p x_{i,j} \leq s_{i,j}^p C$. By summing the previous inequalities over all machines $i \in T(j)$, for any job $j \in \J^{(2)}$, we get: 

\begin{ceqn}
\begin{align*}
    \left(1-\beta \right) f_j(r_{i,j})  r_{i,j}^p &\leq f_j(r_{i,j})  r_{i,j}^p \sum_{i \in T(j)}x_{i,j} \\
    &\leq \left(\sigma^{(p)}_j(T(j))\right)^p C.
\end{align*}
\end{ceqn}
By the above analysis, we get that $\left(\sigma^{(p)}_j(T(j))\right)^p \geq \left(1-\beta \right) \frac{f_j(r_{i,j})}{C} r_{i,j}^p$ and, thus, 
\begin{ceqn}
\begin{align*}
 \sigma^{(p)}_j(T(j)) &\geq \left(\left(1-\beta \right) \frac{f_j(r_{i,j})}{C} \right)^{\frac{1}{p}}r_{i,j} \\
 &\geq \left(\left(1-\beta \right) \frac{f_j(r_{i,j})}{C} \right)^{\frac{1}{p}} \gamma_{j},
\end{align*}
\end{ceqn}
where the last inequality follows by definition of $r_{i,j}$. Finally, by using Fact \ref{fact:technical} and the fact that $\left( \frac{C}{f_j(\gamma_j)} \right)^{\frac{1}{p}} \leq \frac{C}{f_j(\gamma_j)}$ since $p \geq 1$, we get that $f_j(\sigma^{(p)}_j(T(j)) \leq \left(\frac{1}{1-\beta}\right)^{\frac{1}{p}} C$. As previously, the proof of the proposition follows by the fact that, by definition of our algorithm, each machine $i \in \M$ processes at most one job from $\J^{(2)}$.
\end{proof}

It is not hard to see that, given $p$, the algorithm can initially compute a threshold $\beta \in [0,1]$ that minimizes the above theoretical bound. Clearly, for $p=1$ the minimizer of the expression is $\beta = 1/2$, yielding the $4$-approximation of the standard case, while for $p \to +\infty$ one can verify that $\beta \to 1$ and:
\begin{ceqn}
\begin{align*}
\lim_{p \to +\infty} \inf_{\beta \in [0,1]} \left( \frac{1}{\beta} + \frac{1}{(1-\beta)^{1/p}} \right) = 2.
\end{align*}
\end{ceqn} As expected, for the limit case where $p \to +\infty$, our algorithm converges to the well-known algorithm by Lenstra et al. \cite{LST90}, given that our problem becomes non-malleable. By using the standard approximation $\beta = 1 - \frac{\ln (p)}{p}$ for $p \geq 2$, the following theorem follows directly.
\begin{theorem}
Any feasible solution of [LP$^{(p)}$(C)] for $p \geq 2$ can be rounded in polynomial-time into a feasible schedule of makespan at most $\left(\frac{p}{p- \ln (p)} + \sqrt[p]{\frac{p}{\ln(p)}} \right)C$.
\end{theorem}
Note that an analogous approach can handle the case where jobs have different regularizers, with the approximation ratio for this scenario determined by the smallest regularizer that appears in the instance (note that the approximation factor is always at most $4$).

\section{Hardness results and integrality gaps}\label{sec:hardness}
We are able to prove a hardness of approximation result for the case of scheduling malleable jobs on unrelated machines. 

\restateHardness*
\begin{proof}
Following a similar construction as in \cite{CM15} for the case of splittable jobs, we prove the APX-hardness of the malleable scheduling problem on unrelated machines by providing a reduction from the \textsc{max-k-cover} problem: Given a universe of elements $\mathcal{U}=\{e_1, \dots , e_m\}$ and a family of subsets $S_1,\dots, S_n \subseteq \mathcal{U}$, find $k$ sets that maximize the number of covered elements, namely, maximize $\{\left| \bigcup_{i \in I} S_i \right| \suchthat I \subset [n], |I| = k\}$. In \cite{F98}, Feige shows that it is NP-hard to distinguish between instances such that all elements can be covered with $k$ disjoint sets and instances where no $k$ sets can cover more than a $(1 - \frac{1}{e} ) + \epsilon'$ fraction of the elements, for any $\epsilon' > 0$. In addition, the same hardness result holds for instances where all sets have the same cardinality, namely $\frac{m}{k}$. Note that in the following reduction, while we allow for simplicity the speeds to take rational values, the proof stays valid under appropriate scaling.

Given a \textsc{max-k-cover} instance, where each set has cardinality $\frac{m}{k}$, we construct an instance of our problem in polynomial time as follows: We consider $n$ jobs, one for each set $S_j$, and we define the processing time of each job to be $f_{j}(S) = \max\{\frac{1}{\sigma_j(S)}, 1\}$. It is not hard to verify that $f_j(S)$ is non-increasing and $\sigma_j(S)f_j(S)$ is non-decreasing in the total allocated speed. We consider a set $\mathcal{P}$ of $n-k$ {\em common-}machines, such that $s_{i,j} = 1$ for $j \in \J$ and $i \in \mathcal{P}$. Moreover, for every element $e$, we consider an {\em element-}machine $i_e$ such that $s_{i_e,j} = \frac{k}{m}$ if $e \in S_j$ and $s_{i_e,j} = 0$, otherwise. In the following, we fix $\epsilon>0$ such that $\frac{1}{1-\frac{1}{e} + \epsilon'} = \frac{e}{e-1} - \epsilon$.

Consider the case where all elements can be covered by $k$ disjoint sets. Let $\mathcal{C}$ be the family of sets in a cover. In that case, we can assign to each job $j$ such that $S_j \in \mathcal{C}$ the element-machines that correspond to $S_j$. Clearly, every such job allocates $\frac{m}{k}$ machines of speed $s_{i,j} = \frac{k}{m}$, thus, receiving a total speed of one. The rest of the $n-k$ jobs, can be equally distributed to the $n-k$ common-machines, yielding a total makespan of $\OPT = 1$. 
On the other hand, consider the case where no $k$ sets can cover more than a $(1 - \frac{1}{e} ) + \epsilon'$ fraction of the elements. In this case, we can choose any $n-k$ jobs and assign them to the common-machines with processing time exactly $1$. Notice that since $f_j(S) \geq 1$, every (common- or element-) machine can be allocated to at most one job (otherwise the makespan becomes at least $2$). Given that exactly $n-k$ jobs are scheduled on the common machines, we have $k$ jobs to be scheduled on the $m$ element-machines. Aiming for a schedule of makespan at most $\frac{e}{e-1} - \epsilon$, each of the $k$ jobs that are processed by the element-machines should allocate at least $1-\frac{1}{e} + \epsilon'$ speed, that is, at least $\frac{m}{k}(1-\frac{1}{e} + \epsilon')$ element-machines. Therefore, we need at least $m (1 - \frac{1}{e} + \epsilon')$ machines in order to schedule the rest of the jobs within makespan $\frac{e}{e-1} - \epsilon$. However, by assumption on the instance of the \textsc{max-k-cover}, for any choice of $k$ sets, at most $(1 - \frac{1}{e} + \epsilon') m $ machines can contribute non-zero speed, which leads to a contradiction. Therefore, given a $(\frac{e}{e-1} - \epsilon)$-approximation algorithm for the problem of scheduling malleable jobs on unrelated machines, we could distinguish between the two cases in polynomial-time.
\end{proof}

Notice that the above lower bound $\frac{e}{e-1}$ is strictly larger than the well-known $1.5$-hardness for the standard (non-malleable) scheduling problem on unrelated machines. To further support the fact that the malleable version of the problem is harder than its non-malleable counterpart, we provide a pseudopolynomial transformation of the latter to malleable scheduling with speed-implementable processing times. 

\begin{theorem}
There exists a pseudopolynomial transformation of the standard problem of makespan minimization on unrelated machines to the problem of malleable scheduling with speed-implementable processing times.
\end{theorem}
\begin{proof}
Consider an instance of the problem of scheduling non-malleable jobs on unrelated machines. We are given a set of machines $\M$ and a set of jobs $\J$ as well as processing times $p_{i,j} \in \mathbb{Z}_+$ for each $i \in \M$ and each $j \in \J$, with the goal of finding an assignment minimizing the makespan.
We create an equivalent instance of malleable scheduling on unrelated machines on the same set of machines and jobs by defining the processing time functions and speeds as follows:
Let $p_{\max} := \max_{i \in \M, j \in \J} p_{i,j}$.
For $j \in \J$ define $f_j(s) := p_{\max} s^{-\frac{1}{p_{\max}|\J||\M|}}$ and for each $i \in \M$ define $s_{i,j} := \left(\frac{p_{\max}}{p_{i,j}}\right)^{p_{\max}|\J||\M|}$.

Note that $f_j(\{i\}) = p_{i,j}$. 
Furthermore it is easy to verify that the functions $f_j$ fulfill the monotonic workload requirement and that 
$f_j(S) > \min_{i \in S} p_{i, j} - \frac{1}{|\J|}$
for any $S \subseteq \M$.
Therefore, any solution to the non-malleable problem corresponds to a solution of the malleable problem with the same makespan by running each job on the single machine it is assigned to.
Conversely, any solution to the malleable problem induces a solution of the non-malleable problem by running each job only on the fastest machine it is assigned to, increasing the makespan by less than $1$. Since the optimal makespan of the non-malleable instance is integer, an optimal solution of the malleable instance induces an optimal solution of the non-malleable instance.

Note that the encoding lengths of the speed values are pseudopolynomial in the size of the encoding of the original instance. However, by applying standard rounding techniques to the original instance, we can ensure that the constructed instance has polynomial size in trade for a mild loss of precision. 
\end{proof}
Notice that the above reduction can be rendered polynomial by standard techniques, preserving approximation factors with a loss of $1+\varepsilon$. Finally, a simpler version of the above reduction becomes polynomial in the strong sense in the case of restricted identical machines. In brief, for any job $j \in \J$, we define $f_j(s) = p_j \max\{\frac{1}{s},1\}$, where $p_j$ is the processing time of the job in the non-malleable instance. Moreover, we set $s_{i,j} = 1$ for any job $j$ that can be executed on machine $i$ and $s_{i,j} = 0$, otherwise. 

From the side of algorithmic design, we are able to show the following lower bounds on the integrality gap of [LP(C)] for each of the variants we consider:

\begin{theorem}
The integrality gap of [LP(C)] in the case of unrelated machines is lower bounded by $1+ \varphi \approx 2.618$.
\end{theorem}
\begin{proof}
The instance establishing the lower bound follows a similar construction as in~\cite{CM15}.
Note that while in the following proof we allow speeds to take non-integer values, the result also holds for integer-valued speeds under appropriate scaling of the processing time functions. 
We consider a set $\J = \J' \cup \{\hat{j}\}$ of $2k+1$ jobs to be scheduled on a set $\M = \M_A \cup \M_B$ of machines. The set $\J'$ contains $2k$ jobs, each of processing time $f_j(s) = \max\{\frac{1}{s},\frac{\varphi}{2}\}$, where $s$ is the total allocated speed, and these jobs are partitioned into $k$ {\em groups of two}, $J_\ell$ for $\ell \in [k]$. Every group $J_\ell$ of two jobs, is associated with a machine $i^A_\ell$ such that $s_{i^A_\ell,j} = \frac{2}{\varphi}$ for all $j \in J_\ell$ and $s_{i^A_\ell,j} = 0$, otherwise. Let $\M_A$ be the set of these machines and note that $|\M_A| = k$. Moreover, every job of $\J'$ is associated with a {\em dedicated} machine $i^B_j$ such that $s_{i^B_j, j} = 2 - \varphi$, only for job $j$ and $s_{i^B_j, j'} = 0$, otherwise. Let $\M_B$ be the set of these machines. Finally, job $\hat{j}$ has processing time $f_{\hat{j}}(s) = \max\{\frac{1}{s}, 1\}$, while $s_{i,\hat{j}} = 1$ for all $i \in \M_A$ and $s_{i,\hat{j}} = 0$ for all $i \in \M_B$.

In the above setting, it is not hard to verify that the makespan of an optimal solution is $\OPT = 1+\varphi$. Specifically, job $\hat{j}$ uses exactly one machine of $\M_A$ to be executed, given that additional machines cannot decrease its processing time. Let $\hat{i} \in \M_A$ be that machine and let $J_{\hat{\ell}} = \{\hat{j_1}, \hat{j_2}\}$ be the group of two jobs that is associated with $\hat{i}$. Clearly, if at least one of $\hat{j_1}$ and $\hat{j_2}$ is scheduled only on its dedicated machine it is the case that $f_{\hat{j_1}}(2-\varphi) = f_{\hat{j_2}}(2-\varphi) = \frac{1}{2 - \varphi} = 1 + \varphi$. On the other hand, if both $\hat{j_1}$ and $\hat{j_2}$ make use of their $\hat{i}$, then the load of $\hat{i}$ is at least $1 + \varphi$.

Consider the following solution of [LP(C)] for $C = 1 + \frac{1}{k}$. Notice that for $C = 1 + \frac{1}{k}$ and $k \geq 1$, we have that $\gamma_j = \lceil \frac{1}{C} \rceil = \lceil \frac{k}{k+1} \rceil = 1$ for any $j \in \J'$ and $\gamma_{\hat j} = \lceil \frac{1}{C} \rceil = 1$. Therefore, for any $j \in \J'$, we have $r_{i,j} = \max\{\gamma_j, s_{i,j}\} = \frac{2}{\varphi}$, $\forall i \in \M_A$ and $r_{i,j} = 1$, $\forall i \in \M_B$. Finally, we have $r_{i,\hat{j}} = 1$, $\forall i \in \M_A \cup \M_B$.

For each job $j \in \J$ that belongs to the group $J_{\ell}$, we set $x_{i^A_{\ell},j} = \frac{1}{\varphi}$ for the assignment of $j$ to its corresponding machine in $\M_A$ and $x_{i^B_j,j} = \frac{\varphi - 1}{\varphi}$ for the assignment to its dedicated machine of $\M_B$. Moreover, we set $x_{i, \hat{j}} = \frac{1}{k}$ for every $i \in \M_A$. Notice that, for this assignment, constraints \eqref{lp-assign} and \eqref{lp-posit} are trivially satisfied. 

For verifying constraints \eqref{lp-makespan}, for any machine $i \in \M_A$ with corresponding jobs $j_1,j_2 \in \J'$, we have:
\begin{ceqn}
\begin{align*}
f_{\hat{j}}(r_{i,\hat{j}})\frac{r_{i,\hat{j}}}{s_{i,\hat{j}}} x_{i,\hat{j}} &+ f_{j_1}(r_{i,j_1})\frac{r_{i,j_1}}{s_{i,j_1}} x_{i,j_1} + f_{j_2}(r_{i,j_2})\frac{r_{i,j_2}}{s_{i,j_2}} x_{i,j_2} \\
&= x_{i,\hat{j}} + \frac{\varphi}{2} x_{i,j_1} + \frac{\varphi}{2} x_{i,j_2} \\
&= \frac{1}{k} + 1 = C.
\end{align*}
\end{ceqn}
Finally, for any machine $i \in \M_B$ dedicated to a job $j \in \J'$, we have that:
\begin{ceqn}
\begin{align*}
f_{j}(r_{i,j})\frac{r_{i,j}}{s_{i,j}} x_{i,j} = \frac{1}{2-\varphi} x_{i,j} = \frac{\varphi - 1}{\varphi(2-\varphi)} = 1 \leq C,
\end{align*}
\end{ceqn}
where the last equality holds for number $\varphi$. Given the above construction, for $k \to \infty$, there exists an instance such that the integrality gap of [LP(C)] is at least $1 + \varphi$, since:
\begin{ceqn}
\begin{align*}
\lim_{k\to\infty} \frac{\OPT}{C} = \lim_{k\to\infty} \frac{1 + \varphi}{1 + \frac{1}{k}} = 1 + \varphi \approx 2.618.
\end{align*}
\end{ceqn}
\end{proof}

\begin{theorem}
The integrality gap of [LP(C)] in the case of restricted identical machines is lower bounded by $2$.
\end{theorem}
\begin{proof}
We consider a set $\J$ of $k$ identical jobs, each of processing time $f_j(s) = \max\{\frac{2}{s} , 1\}$, where $s$ is the total allocated speed. It is not hard to verify that $f_j$ is monotone non-increasing, while its work $s \max\{\frac{2}{s} , 1\} = \max\{2 , s\}$ is non-decreasing. Every job $i$ can be executed on a {\em dedicated} machine, let $i_j$, while there exists a common pool of $k-1$ machines $\mathcal{P}$ that can be used by any job. Recall that in the restricted assignment case, every feasible pair $(i,j)$ has unit speed. Clearly, the optimal makespan of the above family of instances for $k \geq 1$ is $\OPT = 2$, since there exists at least one job that must be executed only on its dedicated machine. Assuming this is not the case, it has to be that exactly $k$ jobs make use of the common pool of $k-1$ machines, which, by pigeonhole principle, cannot happen within a makespan of at least $\OPT = 2$. 

Consider the following solution of [LP(C)], for $C = \frac{2k}{2k-1} \in (1,2)$: We set $x_{i_j,j} = \frac{k}{2k-1}$ for the assignment of each job to its dedicated machine and $x_{i,j} = \frac{1}{2k-1}$ for every $j \in \J$ and $i \in \mathcal{P}$. According to this assignment, for every job $j \in \J$ we have that $\gamma_j = \lceil \frac{2}{C} \rceil = 1$, while for any machine $i$ we have $r_{i,j} = \max\{s_{i,j}, \gamma_j\} = 1$. Therefore, for the dedicated machine of each job $j$, constraint \eqref{lp-makespan} corresponding to $i_j$ is satisfied with equality since: $f_j(r_{i,j}) r_{i,j} x_{i_j,j} = 2 x_{i_j,j} = \frac{2k}{2k-1} = C$. Moreover, for any machine $i \in \mathcal{P}$, constraints \eqref{lp-makespan} are also satisfied with equality, since,  $\sum_{j \in \J} f_j(r_{i,j}) r_{i,j} x_{i,j} = 2k \frac{1}{2k-1} = C$. Notice, that the above assignment satisfies constraints \eqref{lp-assign} and \eqref{lp-posit}. According to the above construction, for $k \to \infty$, there exists an instance such that the integrality gap of [LP(C)] is at least $2$, since:
\begin{ceqn}
\begin{align*}
\lim_{k\to\infty} \frac{\OPT}{C} = \lim_{k\to\infty} 2 \left(1 - \frac{1}{2k}\right) = 2.
\end{align*}
\end{ceqn} 
Finally, it is not hard to verify that $C = \frac{2k}{2k-1}$ is the smallest $C$ such that [LP(C)] is feasible. Indeed, by summing over constraints \eqref{lp-assign}, we get that $$
k = \sum_{j\in \J} \left(x_{i_j, j} + \sum_{i \in \mathcal{P}} x_{i,j} \right) \leq k\frac{C}{2} + (k-1) \frac{C}{2} \leq (k - \frac{1}{2})C,$$ where the first inequality follows by constraints \eqref{lp-makespan}.
\end{proof}

\begin{theorem}
The integrality gap of [LP(C)] in the case of uniform machines is lower bounded by $2$.
\end{theorem}
\begin{proof}
We consider a set $\J$ of $2k+1$ identical jobs, each of processing time $f_j(s) = \max\{\frac{2}{s} , 1\}$, where $s$ is the total allocated speed. It is not hard to verify that $f_j$ is monotone non-increasing and has non-decreasing work. Moreover, we consider a set $\mathcal{M}_S$ of $2k$ {\em slow} machines of speed $s_i = 1$ for all $i \in \mathcal{M}_S$ and a set $\mathcal{M}_F$ of $k$ {\em fast} machines of speed $s_i = 2$ for all $i \in \mathcal{M}_F$. Clearly, the optimal makespan of the above family of instances for $k \geq 1$ is $\OPT = 2$. By construction of the instance, the optimal makespan is always an integer number. Assuming that $\OPT = 1$, it has to be that every fast machine process at most one job ($k$ jobs in total), while the set of slow machines should process exactly $k$ jobs in pairs of two. By pigeonhole principle, since the number of jobs is $2k+1$, there exist a job that remains to be scheduled and the minimum processing time of this job is one, a contradiction.

Consider the following solution of [LP(C)], for $C = \frac{2k+1}{2k} \in (1,2)$: We set $x_{i,j} = \frac{1}{4k}$ for the assignment of every $j \in \J$ on $i \in \mathcal{M}_S$ and $x_{i,j} = \frac{1}{2k}$ for the assignment of every $j \in \J$ on $i \in \mathcal{M}_F$. For $C = \frac{2k+1}{2k}$, the critical value of any job $j$ becomes $\gamma_j = \lceil \frac{2}{C} \rceil = 2$. Moreover, for any job $j \in \J$ we have $r_{i,j} = \max_{\{\gamma_j,s_i\}} = 1$ for $i \in \mathcal{M}_S$ and $r_{i,j} = \max_{\{\gamma_j,s_i\}} = 2$ for $i \in \mathcal{M}_F$. Clearly, constraints \eqref{lp-assign} and \eqref{lp-posit} of [LP(C)] are satisfied. For this assignment, for every $i \in \mathcal{M}_S$, we have: 
$$\sum_{j \in \J} f_j(r_{i,j}) \frac{r_{i,j}}{s_i} x_{i,j} = 2 \sum_{j \in \J} x_{i,j} = 2 \frac{2k+1}{4k} = C.$$ Moreover, for every $i \in \mathcal{M}_F$, we get: 
$$\sum_{j \in \J} f_j(r_{i,j}) \frac{r_{i,j}}{s_i} x_{i,j} = \sum_{j \in \J} x_{i,j} = \frac{2k+1}{2k} = C.$$
Therefore, constraints \eqref{lp-makespan} are satisfied for every $i \in \mathcal{M}_S \cup \mathcal{M}_F$. Given the above construction, for $k \to \infty$, there exists an instance such that the integrality gap of [LP(C)] is at least $2$, since:
\begin{ceqn}
\begin{align*}
\lim_{k\to\infty} \frac{\OPT}{C} = \lim_{k\to\infty} 2 \left(1 - \frac{1}{2k}\right) = 2.
\end{align*}
\end{ceqn} 
Finally, it is not hard to verify that $C = \frac{2k}{2k-1}$ is the smallest $C$ such that [LP(C)] is feasible. Indeed, by summing constraints \eqref{lp-assign} over all $j \in \J$, we get that: $2k+1 = \sum_{j \in \J} \left( \sum_{i\in \mathcal{M}_S} x_{i, j} + \sum_{i\in \mathcal{M}_S} x_{i,j} \right) \leq 2k\frac{C}{2} + k C = 2k C$, where the inequality follows by using constraints \eqref{lp-makespan}.
\end{proof}

\section{The case of supermodular processing time functions}
\label{sec:supermodular}
In this paper we concentrated our study on speed-imple\-mentable processing time functions.
However, the general definition of malleable scheduling, as given in Section~\ref{sec:intro}, leaves room for many other possible variants of the problem with processing times given by monotone non-increasing set functions.
One natural attempt of capturing the assumption of non-decreasing workload is to assume that, for each job $j \in \mathcal{J}$, the corresponding processing time function $f_j$ is supermodular, i.e.,
$$f_j(T \cup \{i\}) - f_j(T) \geq f_j(S \cup \{i\}) - f_j(S)$$
for all $S \subseteq T \subseteq \mathcal{M}$ and $i \in \mathcal{M} \setminus T$.
The interpretation of this assumption is that the decrease in processing time when adding machine $i$ diminishes the more machines are already used for job $j$ (note that the terms on both sides of the inequality are non-positive because $f_j$ is non-increasing).
For this setting, which we refer to as {\em generalized malleable scheduling with supermodular processing time functions}, we derive a strong hardness of approximation result.
\begin{theorem}\label{super:inapprox}
There is no $|\mathcal{J}|^{1-\varepsilon}$-approximation for generalized malleable scheduling with supermodular processing time functions, unless $P = NP$.
\end{theorem}
\begin{proof}
We show this by reduction from \textsc{graph coloring}: Given a graph $G = (V, E)$, what is the minimum number of colors needed to color all vertices such that no to adjacent vertices have the same color? It is well-known that this problem does not admit a $|V|^{1 - \varepsilon}$-approximation unless $P = NP$~\cite{feige1998zero}.
  
Given a graph $G = (V, E)$, we introduce a job $j_v$ for each $v \in V$ and a machine $i_e$ for each $e \in E$.
For each $v \in V$, let $\delta(v)$ be the set of incident edges and define the corresponding set of machines $S_v := \{i_e : e \in \delta(v)\}$.
We define the processing time function of job $j_v \in \mathcal{J}$ by $$f_{j_v}(S) := 1 + |V| |S_{v} \setminus S|.$$
It is easy to verify that these functions are non-increasing and supermodular.
We show that the optimal makespan for the resulting instance of generalized malleable scheduling is equal to the minimum number of colors needed to color the graph $G$.
  
First assume that $G$ has a coloring with $k$ colors. 
We create a schedule with makespan at most $k$ as follows.
Arbitrarily label the colors $\{0, \dots, k-1\}$ and let $c(v)$ be the color of vertex $v \in V$. 
For each $v \in V$, start job $j_v$ on the set of machines $S_v$ at time $c(v)$.
Because $f_{j_v}(S_v) = 1$, each job $j_v$ is done at time $c(v) + 1$ and two jobs $j_{v}, j_{v'}$ only run in parallel if $c(v) = c(v')$.
Because no two adjacent vertices have the same color, $c(v) = c(v')$ implies $S_v \cap S_{v'} = \emptyset$.
Hence every machine runs at most one job at any given time, which shows that the schedule is feasible.
Its makespan is $k$ as the last job starts at time $k-1$.
  
Now assume there is a schedule with makespan $C$. We show there is a coloring with at most $\lfloor C \rfloor$ colors.
We can assume that $C \leq |V|$, as otherwise the trivial coloring suffices.
Hence, for each $v \in V$, the subset of machines that job $j_v$ is assigned to contains the set $S_v$.
Define the color of vertex $v \in V$ by $c(v) := \lfloor C_{j_v} \rfloor$, i.e., the completion time of the corresponding job rounded down.
Note that $c(v) \in \{1, \dots, \lfloor C \rfloor\}$, because each job has a processing time of at least $1$ and the last job finishes at time $C$.
Furthermore, if $c(v) = c(v')$ for some $v, v' \in V$, then there is a time $t$ where $j_v$ and $j_{v'}$ are both being processed in the schedule. Since each machine in $S_v$ is assigned to job $j_v$ and each machine in $S_{v'}$ is assigned to job $j_{v'}$, we conclude that $S_v \cap S_{v'} = \emptyset$, i.e., $v$ and $v'$ are not adjacent.
Hence the coloring is feasible.

By the above analysis, we conclude that any poly\-nomial-time $|\J|^{1-\epsilon}$-approximation algorithm for the generalized malleable scheduling problem with supermodular processing time functions would imply a $|V|^{1-\epsilon}$-approximation algorithm for \textsc{graph coloring}, which is a contradiction, unless $P = NP$.
\end{proof}
\section*{Conclusion}
In this work, we propose and study a generalization of the malleable scheduling problem in the setting of non-identical machines. For this problem, we design constant approximation algorithms for the cases of unrelated, uniform and restricted identical machines. 
Although our generalized model widens the amount of applications captured comparing to the case of malleable scheduling of identical machines, it does not yet capture issues of interdependence between the machines in an explicit manner. As an example, consider the case where a set of machines performs better in combination due to locality or other aspects.

In this direction, an interesting future work can be the study of the malleable scheduling problem with more general processing time functions that are able to capture the particularities of real-life resource allocation systems.

\bibliographystyle{plainurl}
\bibliography{ref}

\newpage
\appendix
\section{Appendix}

\subsection{A slight improvement by optimizing the threshold}\label{appendix:unrelated:tuning}

Recall that in both algorithms for the unrelated machines case, the threshold for deciding whether a job $j$ is assigned to $p(j)$ or to $T(j)$ is set to $1/2$. While this is the optimal choice for the simple $4$-approximate rounding scheme, we can achieve a slightly better bound for our improved algorithm by optimizing this threshold accordingly. 
Let $\beta \in (0,1)$ be the threshold such that a job $j$ is assigned to $p(j)$ when $x_{p(j),j} \geq \beta$, or to the set $T(j)$, when $\sum_{i \in T(j)} x_{i,j} > 1 - \beta$. 
Formally, let $\J^{(1)} := \{j \in \J \suchthat x_{p(j),j} \geq \beta\}$ be the set of jobs that are assigned to their parent-machines and $\J^{(2)} := \J \setminus \J^{(1)}$ the rest of the jobs.
For $j \in \J^{(2)}$ and $\theta \in [0, 1]$ define $S_j(\theta) := \{ i \in T(j) \suchthat 1 - \frac{\ell_i}{C} \geq \theta\}$. Choose $\theta_j$ so as to minimize $2(1 - \theta_j)C + f_j(S_j(\theta_j))$ (note that this minimizer can be determined by trying out at most $|T(j)|$ different values for $\theta_j$). 
We then assign each job in $j \in J^{(2)}$ to the machine set $S_j(\theta_j)$.

Recall that for any $i \in \M$ there is at most one $j \in J^{(2)}$ with $i \in T(j)$. If $i \notin S_j(\theta_j)$, then load of machine $i$ is bounded by $\frac{1}{\beta} \ell_i \leq \frac{1}{\beta} C$, where $\ell_i$ as defined in Section~\ref{sec:rounding:simple}. If $i \in S_j(\theta_j)$, then the load of machine $i$ is bounded by
\begin{ceqn} 
\begin{align}
\max_{i' \in S_j(\theta_j)}\Big\{\beta^{-1} \ell_{i} + f_j(S_j(\theta_j)) \Big\} \leq \beta^{-1}(1 - \theta_j)C + f_j(S_j(\theta_j)),
\end{align}
\end{ceqn}
where the inequality comes from the fact that $1 - \frac{\ell_{i'}}{C} \geq \theta_j$ for all $i' \in S_{\theta_j}$. 

We now fix $\beta$ to be the unique solution of $\beta^{-1}+1 = \frac{e^{\frac{1}{\beta} - 1}}{\beta(e^{\frac{1}{\beta} - 1} - 1)}$ in the interval $[0,1]$. The following proposition gives an upper bound on the RHS of \eqref{speed:load} as a result of our filtering technique.

\begin{proposition}
For each $j \in \J^{(2)}$, there is a $\theta \in [0,1]$ with $\beta^{-1}(1- \theta)C + f_j(S_j(\theta)) \leq \frac{e^{\frac{1}{\beta} - 1}}{\beta(e^{\frac{1}{\beta} - 1} - 1)}C$.
\end{proposition}

\begin{proof}
We first assume that for all $i \in T(j)$, it is the case that $s_{i,j} \leq \gamma_j$ and, thus, $r_{i,j} = \gamma_j$. In the opposite case, where there exists some $i' \in T(j)$ such that $s_{i,j} > \gamma_j$, by choosing $\theta = 0$, then \eqref{speed:load} can be upper bounded by $\beta^{-1}C + f_j(S_j(0)) = \beta^{-1}C + f_j(T(j)) \leq \left(\beta^{-1}+1\right)C$. In that case, the proposition follows directly by the fact that $\beta^{-1}+1 = \frac{e^{\frac{1}{\beta} - 1}}{\beta(e^{\frac{1}{\beta} - 1} - 1)}$, by our choice of $\beta$.

Define $\alpha := \frac{e^{\frac{1}{\beta} - 1}}{\beta(e^{\frac{1}{\beta} - 1} - 1)}$. We show that there is a $\theta \in [0, 1]$ with $\sigma_j(S_j(\theta)) \geq \frac{\gamma_j f_j(\gamma_j)}{(\alpha + \beta^{-1} \theta -\beta^{-1})C}$. Then $f_j(S_j(\theta)) \leq (\alpha + \beta^{-1} \theta -\beta^{-1}) C$ by Fact~\ref{fact:technical}, implying the lemma.

Define the function $g: [0,1] \rightarrow \mathbb{R}_{+}$ by $g(\theta) := \sigma_j (S_j(\theta))$. It is easy to see $g$ is non-increasing integrable and that 
 $$\int_{0}^{1} g(\theta) d\theta = \sum_{i \in T(j)} s_{i,j} (1 - \frac{\ell_i}{C}).$$ 
 
Now assume by contradiction that $g(\theta) < \frac{\gamma_j f_j(\gamma_j)}{(\alpha + \beta^{-1} \theta - \beta^{-1})C}$ for all $\theta \in [0, 1]$.
Note that $\ell_i + \frac{\gamma_j f_j(\gamma_j)}{s_{i,j}} x_{i,j} \leq C$ for every $i \in T(j)$ by constraints \eqref{lp-makespan}.
Hence  $\frac{f_j(\gamma_j)\gamma_j}{C} x_{i,j} \leq s_{i,j}(1 - \frac{\ell_i}{C})$ for all $i \in T(j)$.
Summing over all $i \in T(j)$ and using the fact that $\sum_{i \in T(j)} x_{i,j} \geq 1 - \beta$ because $j \in J^{(2)}$, we get
\begin{ceqn}
\begin{align*}
(1- \beta)\frac{f_j(\gamma_j)\gamma_j}{C} &\leq \sum_{i \in T(j)} s_{i,j} (1 - \frac{\ell_i}{C})\\ 
&= \int_{0}^{1} g(\theta) d\theta \\
&< \frac{f_j(\gamma_j)\gamma_j}{C} \int_{0}^{1} \frac{1}{\alpha + \beta^{-1} \theta - \beta^{-1}} d\theta,
\end{align*}
\end{ceqn}
where the last inequality uses the assumption that $g(\theta) < \frac{\gamma_j f_j(\gamma_j)}{(\alpha + \beta^{-1} \theta_j - \beta^{-1})C}$ for all $\theta \in [0,1]$. By simplifying the above inequality, we get the contradiction
\begin{ceqn}
\[\frac{1 - \beta}{\beta} < \int_{\alpha-\beta^{-1}}^{\alpha} \frac{1}{\lambda} d\lambda = \ln(\frac{\alpha}{\alpha -\beta^{-1}}) = \frac{1 - \beta}{\beta}.\]
\end{ceqn}
\end{proof}
Therefore, by choosing $\alpha = \inf_{\beta \in (0,1)} \{ \frac{e^{\frac{1}{\beta} - 1}}{\beta(e^{\frac{1}{\beta} - 1} - 1)} \} \approx 3.14619$ \footnote{Note that the minimizer of this expression coincides with the unique solution of $\beta^{-1}+1 = \frac{e^{\frac{1}{\beta} - 1}}{\beta(e^{\frac{1}{\beta} - 1} - 1)}$ in the interval $[0,1]$.}, with threshold $\beta \approx 0.465941$, we can prove the following theorem.
\begin{theorem}
There exists a polynomial-time $3.1461$-approximation algorithm for the problem of scheduling malleable jobs on unrelated machines.
\end{theorem}

\end{document}